\documentclass[11pt,english]{article}
\usepackage{ae,aecompl}
\usepackage[T1]{fontenc}
\usepackage[latin9]{inputenc}
\usepackage{geometry}
\usepackage{amsmath}
\usepackage{amsthm}
\usepackage{amssymb}
\usepackage{graphicx}
\usepackage{setspace}
\usepackage[authoryear,round]{natbib}

\makeatletter
\theoremstyle{plain}
\newtheorem{thm}{\protect\theoremname}
\theoremstyle{plain}
\newtheorem{cor}{\protect\corollaryname}
\theoremstyle{plain}
\newtheorem{lem}{\protect\lemmaname}
\theoremstyle{plain}
\newtheorem{prop}{\protect\propositionname}
\theoremstyle{plain}
\newtheorem*{theoremprime}{Theorem 1$'$}
\theoremstyle{definition}
 \newtheorem{example}{\protect\examplename}

\usepackage{setspace}
\usepackage{pgf}
\usepackage{tikz}

\DeclareMathOperator{\Supp}{Supp}

\DeclareMathOperator*{\argmax}{arg\,max} 
\newcommand{\bdu}{\mathbf{u}}

\usepackage{pgf}
\usepackage{tikz}
\usetikzlibrary{patterns}
\usetikzlibrary{decorations.markings}
\usepackage{caption}

\usepackage{graphicx}

\usepackage{makecell}
\newcolumntype{C}{>{\centering\arraybackslash $}p{1.3cm}<{$}}
\newcolumntype{D}{>{\centering\arraybackslash $}p{1.4cm}<{$}}
\newcolumntype{S}{>{\centering\arraybackslash $}p{1.8cm}<{$}}
\newcolumntype{N}{>{\centering\arraybackslash $}p{1.6cm}<{$}}

\date{}

\makeatother

\usepackage{babel}
\providecommand{\corollaryname}{Corollary}
\providecommand{\examplename}{Example}
\providecommand{\lemmaname}{Lemma}
\providecommand{\propositionname}{Proposition}
\providecommand{\theoremname}{Theorem}

\begin{document}
\title{Commitment and Randomization in Communication}
\author{Emir Kamenica and Xiao Lin\thanks{Kamenica: Booth School of Business, University of Chicago, 5807 S.
Woodlawn Ave., Chicago, IL 60637 (emir@uchicago.edu); Lin: Department
of Economics, University of Pennsylvania, 133 S. 36th St., Philadelphia
PA 19104 (xiaolin7@sas.upenn.edu). We are grateful to James Best,
Nina Bobkova, Ben Brooks, Nima Haghpanah, Ilwoo Hwang, Yuhta Ishii,
Jonathan Libgober, Elliot Lipnowski, Eric Mbakop, John Moore, Roger
Myerson, Phil Reny, Daniel Rappoport, Vasiliki Skreta, Ina Taneva,
and Rakesh Vohra for helpful comments.} \\
 \\
November 2025}
\maketitle
\begin{abstract}
\begin{doublespace}
When does Sender, in a Sender-Receiver game, strictly value commitment?
In a setting with finitely many actions and states, we establish that,
generically, commitment has no value if and only if a partitional
experiment is optimal. Moreover, if Sender's preferred cheap-talk
equilibrium necessarily involves randomization, then Sender values
commitment. Our results imply that if a school values commitment to
a grading policy, then the school necessarily prefers to grade unfairly.
We also ask: for what share of preferences does commitment have no
value? For any state space, if there are $\left|A\right|$ actions,
the share is at least $\frac{1}{\left|A\right|^{\left|A\right|}}$.
As the number of states grows large, the share converges precisely
to $\frac{1}{\left|A\right|^{\left|A\right|}}$. 
\end{doublespace}

\begin{onehalfspace}
\vspace*{1cm}
\noindent\textit{Keywords}: Bayesian persuasion; cheap talk

\noindent\textit{JEL\ Codes}: D80, D83
\end{onehalfspace}

\vspace*{\fill}
\end{abstract}
\pagebreak{}

\setstretch{1.65}

\section{Introduction}

Commitment is often valuable. In the context of communication, this
fact is brought out by the contrast of Sender's payoff in Bayesian
persuasion versus cheap talk. For any prior, and any profile of Sender
and Receiver's preferences, Sender's payoff is always weakly higher
under Bayesian persuasion than in any cheap-talk equilibrium.\footnote{In fact, Bayesian persuasion provides the upper bound on Sender's
equilibrium payoff under any communication protocol, such as disclosure
or signaling; see Proposition 1 in the Online Appendix of \citet{kamenicagentzkow2011}.} In this paper, we ask: when does commitment make Sender \emph{strictly}
better off?

Answering this question would contribute to our understanding of circumstances
that incentivize building strong institutions that are immune to influence
(\citealt{north1993institutions}; \citealt*{lipnowski2022persuasion})
or building a reputation for a degree of honesty (\citealt{best2024persuasion};
\citealt*{mathevet2024reputation}). 

We focus exclusively on environments with finitely many states and
actions. We show that, generically, Sender with commitment values
that commitment if and only if he values randomization (Theorem \ref{Thm:commitmentWTA}).\footnote{Most of the results in this paper hold generically, i.e., for all
but a knife-edge set of preferences. To ease exposition, we henceforth
sweep the genericity modifier under the rug in the Introduction.} In other words, the Bayesian persuasion payoff is achievable in a
cheap-talk equilibrium if and only if a partitional experiment is
a solution to the Bayesian persuasion problem. Moreover, if Sender's
preferred equilibrium in a cheap-talk game necessarily involves randomization,
then Sender values commitment (Theorem \ref{Thm:commitmentWTP}). 

For an application of these results, consider a school that assigns
grades to students, each of whom is characterized by a vector of attributes.
Some of the attributes are relevant, in the sense that an employer
values those attributes or the school's value of placing a student
depends on them. Other attributes are irrelevant. The school assigns
a grade to each student based on her attributes. The school's grading
policy is \emph{fair }if it assigns the same grade to students with
identical relevant attributes. Theorem \ref{Thm:commitmentWTA} tells
us that if the school values committing to a grading policy of any
form (such as mandating a maximum GPA), then the school prefers to
grade unfairly. Conversely, if a fair grading scheme is optimal, there
is no need for commitment: discretionary ``cheap-talk'' grades are
as effective as those disciplined by a publicly declared grading policy. 

We also derive results about the share of preferences such that Sender
finds commitment (or, equivalently, randomization) valuable. We uncover
a potentially surprising connection between the share of preferences
where commitment has value and the cardinalities of the action set
and the state space (Theorem \ref{Thm:how-often}).

Say that preferences are felicitous if Sender can obtain his ideal
payoff under cheap talk. If preferences are felicitous, commitment
obviously has no value. We show that, for any number of states, the
share of felicitous preferences is at least $\frac{1}{\left|A\right|^{\left|A\right|}}$,
where $\left|A\right|$ denotes the cardinality of the action set.
Hence, in any state space, the share of preferences such that commitment
has no value is at least $\frac{1}{\left|A\right|^{\left|A\right|}}$.

Moreover, as the number of states grows large: (i) the share of felicitous
preferences converges to $\frac{1}{\left|A\right|^{\left|A\right|}}$,
and (ii) commitment has no value if and only if preferences are felicitous.
Thus, when the state space is large, the share of preferences such
that commitment has no value is approximately $\frac{1}{\left|A\right|^{\left|A\right|}}$.
For example, if the action set is binary and there are many states,
commitment is valuable for roughly 75\% of possible preferences.

\subsection*{Illustrative example}

The workhorse example in the Bayesian-persuasion literature is a prosecutor
(Sender) trying to convince a judge (Receiver) to convict a defendant
who is guilty or innocent. The judge's preferences are such that she
prefers to convict if the probability of guilt is weakly higher than
the probability of innocence. The prosecutor has state-independent
preferences and always prefers conviction. The prior probability of
guilt is $0.3$. 

If the environment is cheap talk, the unique equilibrium outcome is
that the judge ignores the prosecutor and always acquits the defendant.
If the prosecutor can commit to an experiment about the state, however,
he will conduct a stochastic experiment that indicates guilt whenever
the defendant is guilty and indicates guilt with probability $\frac{3}{7}$
when the defendant is innocent (\citealt{kamenicagentzkow2011}).
This experiment induces the judge to convict the defendant with 60\%
probability. The prosecutor is thus strictly better off than under
cheap talk.

Our Theorem \ref{Thm:commitmentWTA} tells us that the two facts,
(i) the prosecutor's optimal experiment involves randomization and
(ii) the prosecutor does better under commitment, imply each other.\footnote{Theorem \ref{Thm:commitmentWTA} only states that (i) and (ii) imply
one another for a generic set of preferences. To apply the theorem
here, we note that the preferences in the prosecutor-judge example
belong to the generic set used in the proof of the Theorem. Moreover,
in Online Appendix \ref{subsec:transparent} we show that Theorem
\ref{Thm:commitmentWTA} holds when Sender has state-independent preferences.} Of course, the prosecutor-judge example was designed to be extremely
simple, so in this particular example one can easily determine the
optimal experiment and the value of commitment without our result.
In more complicated environments, however, Theorem \ref{Thm:commitmentWTA}
can simplify the determination of whether commitment is valuable.
Except in certain cases, such as uniform-quadratic (\citealt{crawford1982strategic})
or transparent preferences (\citealt{lipnowski2020cheap}), cheap-talk
games can be difficult to solve. Theorem \ref{Thm:commitmentWTA}
can then be used to determine whether commitment is valuable without
solving for cheap talk equilibria, simply by computing the Bayesian-persuasion
optima and checking whether they include a partitional experiment.\footnote{Recent research provides a large toolbox for solving Bayesian-persuasion
problems, including concavification \citep{kamenicagentzkow2011},
price-theoretic approaches (\citealt{kolotilin2018optimal}; \citealt{dworczak2019simple}),
duality (\citealt{dworczak2024persuasion}), and optimal-transport
theory (\citealt*{kolotilin2023persuasion}). \citet{bergemann2016bayes}
show that persuasion problem can be formulated as a linear program;
it is well known that linear programs can be computed in polynomial
time. In contrast, \citet{babichenko2023algorithmic} establish that
it is NP-hard to approximate Sender's maximum payoff in cheap-talk,
or even to determine if that payoff is strictly greater than in a
babbling equilibrium. For a survey of computational approaches to
Bayesian persuasion, see \citet{dughmi2017algorithmic}.} 

The prosecutor-judge example also illustrates the distinction between
the if-and-only-if result in Theorem \ref{Thm:commitmentWTA} and
the unidirectional Theorem \ref{Thm:commitmentWTP}. Recall that Theorem
\ref{Thm:commitmentWTP} does not claim that the value of commitment
is positive only if randomization is valuable in cheap talk. The prosecutor-judge
example provides a counterexample to such a claim. In the cheap-talk
game, the prosecutor has no value for randomization: with or without
it, he never obtains any convictions. Yet, the prosecutor obviously
values commitment.

Finally, the prosecutor-judge example also helps illustrate what Theorem
\ref{Thm:commitmentWTA} does \emph{not} say. Prohibiting randomization
would not mean commitment is not valuable. Suppose that the prosecutor
is endowed with commitment, but is legally obliged to use only partitional
experiments. In that case, the prosecutor would provide a fully informative
experiment, obtaining a conviction with 30\% probability. That is
still better than his cheap-talk payoff of no convictions.

\subsection*{Related literature}

Our paper connects the literatures on cheap talk (\citealt{crawford1982strategic})
and Bayesian persuasion (\citealt{kamenicagentzkow2011}). \citet{min2021bayesian}
and \citet*{lipnowski2022persuasion} examine environments with limited
commitment that are a mixture of cheap talk and Bayesian persuasion.
In contrast, we focus on the question of when cheap talk and Bayesian
persuasion yield the same payoff to Sender.\footnote{\citet{perez2014interim} and \citet{koessler2023informed} examine
the circumstances under which Sender attains his Bayesian persuasion
payoff even if he learns the state prior to selecting the experiment.}

\citet{glazer2006study} and \citet{sher2011credibility} consider
disclosure games and derive conditions on preferences that imply that
Receiver values neither commitment nor randomization. \citet{arieli2025communicationleadefficiency}
establish connections between randomization and efficiency under both
cheap talk and Bayesian persuasion.

Several papers examine value of commitment under the assumption that
Sender has state-independent preferences. When the action space is
finite, as in our framework, \citet{lipnowski2020cheap} show that
(for almost every prior) Sender either: (i) obtains his ideal payoff
in cheap talk, or (ii) values commitment; \citet{best2024persuasion}
show that (for almost every prior) Sender either: (i) obtains his
ideal payoff under the prior, or (ii) values randomization. \citet{titovazhang2024persuasion}
establish a connection between randomization and the attainability
of the Bayesian persuasion payoff under verifiable messages. \citet{corrao2023bounds}
examine Sender's payoff under cheap talk, mediation, and Bayesian
persuasion. They establish that Sender does not value commitment if
his payoffs are the same under mediation and Bayesian persuasion.

In the context of mechanism design, value of commitment and value
of randomization have been studied separately. Mechanism design with
limited commitment has been studied by \citet{akbarpourLi} and \citet{doval2022mechanism},
among others. Value of randomization in mechanism design has been
widely recognized in single-agent multi-product monopolist settings
(e.g., \citealt{manelli2006bundling}). In contrast, with two or more
agents, \citet{chen2019equivalence} establish that if agents' types
are atomless and independently distributed, randomization is never
valuable.

\section{Set-up and definitions}

\subsubsection*{Preference and beliefs}

Receiver (she) has a utility function $u_{R}\left(a,\omega\right)$
that depends on her action $a\in A$ and the state of the world $\omega\in\Omega$.
Both $A$ and $\Omega$ are finite; our analysis relies heavily on
this assumption.\footnote{At the risk of being excessively philosophical, we consider environments
with finite $A$ and $\Omega$ to be more realistic; the use of infinite
sets often provides tractability but rarely improves realism. We discuss
the role of the finiteness assumption in Online Appendix \ref{subsec:finite}.} For any finite set $X$, we denote its cardinality by $\left|X\right|$.
Sender (he) has a utility function $u_{S}\left(a,\omega\right)$ that
depends on Receiver's action and the state. The players share an interior
common prior $\mu_{0}$ on $\Omega$. We say action $a^{*}$ is $i$'s
\emph{ideal action in $\omega$} if $a^{*}\in\argmax_{a\in A}u_{i}(a,\omega)$.

\subsubsection*{Environments, shares, and genericity}

We refer to the pair $\left(u_{S},u_{R}\right)$ as the (preference)\emph{
environment}. 

Since $u_{S}$ and $u_{R}$ have a finite domain, they are bounded.
We further restrict our attention to environments where $u_{S}$ and
$u_{R}$ take values in some fixed interval, which, without loss of
generality, we set to $\left[0,1\right]$. Under these assumptions,
the set of all environments is $\left[0,1\right]^{2\left|A\right|\left|\Omega\right|}$.
When we say that a claim holds for a $\gamma$ share of environments,
we simply mean that the set of environments where the claim holds
has Lebesgue measure $\gamma$ on $\mathbb{R}{}^{2\left|A\right|\left|\Omega\right|}$.

We say a set of environments is \emph{generic} if it has Lebesgue
measure one on $\mathbb{R}{}^{2\left|A\right|\left|\Omega\right|}$.\footnote{Our results also hold if we use a topological rather than a measure-theoretic
notion of genericity. See Footnotes \ref{fn:puq-topology} and \ref{fn:scant-topology}
in the Appendix.} When we say that a claim holds \emph{generically}, we mean that it
holds for a generic set of environments.\footnote{\citet{lipnowski2020equivalence}, who focuses on finite action and
state spaces as we do, establishes that commitment has no value when
Sender's value function over Receiver's beliefs is continuous. Such
continuity, however, holds for a zero share of environments. In contrast,
we focus on results that hold generically.} When we say that, given $A$, a claim holds \emph{generically as
$\left|\Omega\right|\rightarrow\infty$}, we mean that the share of
environments where the claim does not hold converges to zero as $\left|\Omega\right|\rightarrow\infty$.

\subsubsection*{Cheap talk, Bayesian persuasion, and value of commitment}

Let $M$ be a finite message space with $|M|>\max\{|\Omega|,|A|\}$.\footnote{Our results concern Sender's payoffs under cheap talk, Bayesian persuasion,
and restriction to partitional strategies in those models. To derive
Sender's maximal payoff, it is without loss of generality to set $|M|\geq|\Omega|$
for cheap talk (\citealt{Matthews90}), $|M|\geq\min\{|\Omega|,|A|\}$
for Bayesian persuasion (\citealt{kamenicagentzkow2011}), and $|M|\geq|\Omega|$
for partitional strategies (trivially). Therefore, assuming $|M|\geq|\Omega|$
would suffice for our results. However, further assuming $|M|\geq|A|+1$
simplifies the proofs of Lemmas \ref{lem:Receiver-not-mix} and \ref{lem:permissive-implies-obedient}.} Sender chooses a messaging strategy $\sigma:\Omega\rightarrow\Delta M$.
Receiver chooses an action strategy $\rho:M\rightarrow\Delta A$. 

A profile of strategies $(\sigma,\rho)$ induces expected payoffs
\[
U_{i}(\sigma,\rho)=\sum_{\omega,m,a}\,\mu_{0}(\omega)\,\sigma(m|\omega)\,\rho(a|m)\,u_{i}(a,\omega)\quad\text{for }i=S,R.
\]

A profile $(\sigma^{*},\rho^{*})$ is\emph{ S-BR} if $\sigma^{*}\in\argmax_{\sigma}U_{S}(\sigma,\rho^{*})$.
A profile $(\sigma^{*},\rho^{*})$ is\emph{ R-BR} if $\rho^{*}\in\argmax_{\rho}U_{R}(\sigma^{*},\rho)$.

Sender's \emph{ideal payoff} is the maximum $U_{S}$ induced by any
profile.

A \textit{cheap-talk equilibrium} is a profile that satisfies S-BR
and R-BR.\footnote{This definition may seem unconventional since it uses Nash equilibrium,
rather than perfect Bayesian equilibrium, as the solution concept.
In cheap-talk games, however, the set of equilibrium outcomes (joint
distributions of states, messages, and actions) is exactly the same
whether we apply Nash or perfect Bayesian as the equilibrium concept.
The formulation in terms of Nash equilibria streamlines the proofs.} We define (Sender's)\emph{ cheap-talk payoff} as the maximum $U_{S}$
induced by a cheap-talk equilibrium.\footnote{Throughout, we examine the value of commitment to Sender; hence the
focus on Sender's payoff. The set of equilibrium payoffs is compact
so a maximum exists. We are interested in whether Sender can attain
his commitment payoff in \emph{some} equilibrium, so we focus on Sender-preferred
equilibria. Unless no information is the commitment optimum, it cannot
be that \emph{every} cheap-talk equilibrium yields the commitment
payoff since every cheap-talk game admits a babbling equilibrium.
(And even if no information is the commitment optimum, there might
be cheap-talk equilibria that yield a lower payoff to Sender than
babbling.)}

A \textit{persuasion profile} is a profile that satisfies R-BR. The
(Bayesian)\emph{ persuasion payoff} is the maximum $U_{S}$ induced
by a persuasion profile.\footnote{\citet*{lipnowski2024perfect} establish that, with finite $A$ and
$\Omega$, this is generically the only payoff that Sender could attain
in an equilibrium of a Bayesian persuasion game.} We refer to a persuasion profile that yields the persuasion payoff
as \emph{optimal}.

We say that\emph{ commitment is valuable} if the persuasion payoff
is strictly higher than the cheap-talk payoff. Otherwise, we say \emph{commitment
has no value}.

\subsubsection*{Partitional strategies and value of randomization}

A messaging strategy $\sigma$ is \emph{partitional} if for every
$\omega$, there is a message $m$ such that $\sigma\left(m|\omega\right)=1$.
A profile $\left(\sigma,\rho\right)$ is a \emph{partitional profile}
if $\sigma$ is partitional.\footnote{Our focus is on the connection between Sender's value of commitment
and Sender's randomization. Consequently, the definition of a partitional
profile only concerns Sender's strategy. That said, along the way
we will establish a result about Receiver playing pure strategies
(see Lemma \ref{lem:Receiver-not-mix}).} The \emph{partitional persuasion payoff} is the maximum $U_{S}$
induced by a partitional persuasion profile. The \emph{partitional
cheap-talk payoff} is the maximum $U_{S}$ induced by a partitional
cheap-talk equilibrium.\footnote{A partitional cheap-talk equilibrium always exists because the babbling
equilibrium outcome can be supported by Sender always sending the
same message.}

We say that \emph{committed Sender values randomization }if the persuasion
payoff is strictly higher than the partitional persuasion payoff.
We say that \emph{cheap-talk Sender values randomization} if the cheap-talk
payoff is strictly higher than the partitional cheap-talk payoff.

\section{Value of commitment and randomization under commitment}

In this section, we consider a Sender with commitment power, who can
choose his messaging strategy prior to being informed of the state.
We ask whether this commitment power makes Sender strictly better
off. We link the value of commitment to Sender's behavior under commitment,
in particular to whether Sender has a strict preference for randomization.
\begin{thm}
\label{Thm:commitmentWTA}Generically, commitment is valuable if and
only if committed Sender values randomization.
\end{thm}
Here we provide an intuition about the only-if direction of the theorem.
We postpone the discussion of the converse until the next section,
as the intuition for it is related to that for Theorem \ref{Thm:commitmentWTP}.
Formal proofs are in the Appendix.\footnote{Theorem \ref{Thm:commitmentWTA} can be extended to establish a threefold
equivalence. Generically, the following imply each other: (i) commitment
is valuable, (ii) committed Sender values randomization, and (iii)
any optimal persuasion profile induces a belief under which Receiver
has multiple optimal actions (see Theorem 1$'$ in the Appendix).}

For any $E\subseteq\Omega$, let $\mu_{E}$ denote the posterior belief
induced by learning that $\omega$ is in $E$. For a generic set of
environments, Receiver's optimal action given any such $\mu_{E}$
is unique and remains optimal in a neighborhood of beliefs around
$\mu_{E}$. 

Now, suppose that there is a partitional optimal persuasion profile
$\left(\sigma,\rho\right)$. Let $M_{\sigma}$ be the set of messages
that are sent under $\sigma$. Because $\sigma$ is partitional, each
$m\in M_{\sigma}$ is associated with a subset of the state space,
namely $\Omega_{m}\equiv\left\{ \omega|\sigma\left(m|\omega\right)=1\right\} $.
For each $m\in M_{\sigma}$, let $\mu_{m}$ be the belief induced
by $m$, and let $a_{m}$ be Receiver's (uniquely) optimal action
given $\mu_{m}$. As noted above, $a_{m}$ remains optimal in a neighborhood
of beliefs around $\mu_{m}$.

Key to the proof is to note that every action $a_{m}$ taken in equilibrium
must be Sender's preferred action, among the actions taken in equilibrium,
in all states where action $a_{m}$ is taken. In other words, let
$A^{*}=\left\{ a_{m}|m\in M_{\sigma}\right\} $; for every $m\in M_{\sigma}$,
for every $\omega\in\Omega_{m}$, we have $u_{S}\left(a_{m},\omega\right)\geq u_{S}\left(a_{m'},\omega\right)$
for all $a_{m'}\in A^{*}$. Why does this hold? If it were not the
case, Sender could attain a higher payoff with an alternative strategy:
if $u_{S}\left(a_{m},\omega\right)<u_{S}\left(a_{m'},\omega\right)$
for some $a_{m'}\in A^{*}$, $\omega\in\Omega_{m}$, sender could
send $m'$ in $\omega$ with a small probability and still keep $a_{m}$
optimal given $m$. 

Finally, the fact that for every $m\in M_{\sigma}$, $u_{S}\left(a_{m},\omega\right)\geq u_{S}\left(a_{m'},\omega\right)$
for all $a_{m'}\in A^{*}$ and all $\omega\in\Omega_{m}$ implies
that $\left(\sigma,\rho\right)$ is a cheap-talk equilibrium.\footnote{Deviating to an on-path message $\hat{m}\in M_{\sigma}$ cannot be
profitable by the inequality $u_{S}\left(a_{m},\omega\right)\geq u_{S}\left(a_{\hat{m}},\omega\right)$
for $\hat{m}\in M_{\sigma}$; for any off-path message $\hat{m}\notin M_{\sigma}$,
we can just set $\text{\ensuremath{\rho}\ensuremath{\left(\cdot|\hat{m}\right)}}=\ensuremath{\rho}\left(\cdot|m^{*}\right)$
for some $m^{*}\in M_{\sigma},$ thus ensuring that such a deviation
is also not profitable.} Hence, commitment is not valuable.

Theorem \ref{Thm:commitmentWTA} only tells us that, generically,
commitment has \emph{zero }value if and only if randomization has
\emph{zero }value. A natural question is whether, generically, small
value of commitment implies or is implied by small value of randomization.
The answer is no. We construct a positive measure of environments
where the value of randomization is arbitrarily small but the value
of commitment is not (Online Appendix \ref{subsec:Large-comm-small-rand}),
and a positive measure of environments where the value of commitment
is arbitrarily small but the value of randomization is not (Online
Appendix \ref{subsec:Small-comm-large-rand}).

\section{Value of commitment and randomization in cheap talk}

In this section, we consider a Sender without commitment power who
engages in a cheap-talk game. We ask whether he would be strictly
better off if he had commitment power. We link the value of such commitment
to Sender's behavior in Sender-preferred cheap-talk equilibria, in
particular to whether Sender necessarily randomizes in such equilibria.
\begin{thm}
\label{Thm:commitmentWTP}Generically, commitment is valuable if cheap-talk
Sender values randomization.
\end{thm}
Theorem \ref{Thm:commitmentWTP} and the if-direction of Theorem \ref{Thm:commitmentWTA}
both derive from the following result. Generically, if a cheap-talk
equilibrium yields the persuasion payoff, then there is a partitional
$\sigma$ and a (pure-strategy) $\rho$ such that $(\sigma,\rho)$
is a cheap-talk equilibrium and yields the persuasion payoff. We build
this result (Proposition \ref{Lemma:purity} in Appendix \ref{subsec:purity})
in two steps.

The first step (Lemma \ref{lem:Receiver-not-mix}) shows that, generically,
if $\left(\sigma,\rho\right)$ is R-BR and yields the persuasion payoff,
then $\rho$ must be pure on-path. Consider toward contradiction that
there is a message $m$ sent with positive probability under $\sigma$,
and there are two distinct actions, say $a$ and $a'$, in the support
of $\rho\left(\cdot|m\right)$. It must be that both Sender and Receiver
are indifferent between $a$ and $a'$ under $\mu_{m}$, the belief
induced by $m$: Receiver has to be indifferent because $\left(\sigma,\rho\right)$
is R-BR; Sender has to be indifferent because $\left(\sigma,\rho\right)$
yields the persuasion payoff, which maximizes $U_{S}$ over all persuasion
profiles.\footnote{If Sender strictly prefers one action over the other, say $a$ over
$a'$, at $\mu_{m}$, then Sender would obtain a higher payoff if
Receiver always takes $a$ following $m$ (which would remain R-BR
given Receiver's indifference).} The result then follows from establishing that such a coincidence
of indifferences generically cannot arise when Sender is optimizing.
For some intuition for why this is the case, consider Figure \ref{fig:Indifference-optimality}
which illustrates this result when there are three states. Suppose
$a_{1}$ and $a_{2}$ are in the support of $\rho\left(\cdot|m\right)$.
Region $R_{i}$ denotes beliefs where Receiver prefers $a_{i}$. Region
$S_{i}$ denotes beliefs where Sender prefers $a_{i}$. Generically,
the border between $R_{1}$ and $R_{2}$ is distinct from the border
between $S_{1}$ and $S_{2}$ and thus the two borders have at most
one intersection, $\mu_{m}$. Moreover, generically $\mu_{m}$ (if
it exists) is an interior belief. But now, Sender could deviate to
an alternate strategy that induces beliefs $\mu_{1}$ and $\mu_{2}$
instead of $\mu_{m}$, with Receiver still indifferent between $a_{1}$
and $a_{2}$ at both $\mu_{1}$ and $\mu_{2}$. Suppose that Receiver
takes action $a_{i}$ following belief $\mu_{i}$. This strategy is
still R-BR for Receiver and gives Sender a strictly higher payoff.
Thus, we have reached a contradiction. With more than three states
and more than two actions, the proof that the coincidence of indifferences
generically cannot arise is conceptually similar but notationally
more involved. It is presented in the Appendix as Lemma \ref{lemma:full-rank-genericity}.

\begin{figure}
\begin{centering}
\includegraphics{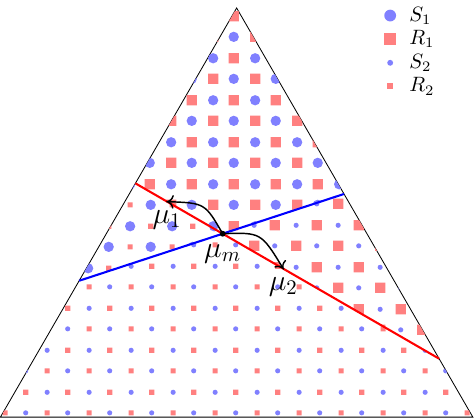}
\par\end{centering}
\caption{\protect\label{fig:Indifference-optimality}Indifference incompatible
with optimality}
\end{figure}

The second step (Lemma \ref{lem:Sender-not-mix}) shows that, generically,
if $\left(\sigma,\rho\right)$ is a cheap-talk equilibrium that yields
the persuasion payoff, and $\rho$ is a pure strategy on-path, then
there is a partitional cheap-talk equilibrium that yields the persuasion
payoff. This is easy to see. Generically, for any $\omega$ and any
$a\neq a'$, we have $u_{S}\left(a,\omega\right)\neq u_{S}\left(a',\omega\right)$.
Now, consider some cheap-talk equilibrium $\left(\sigma,\rho\right)$,
with $\rho$ pure on-path, that yields the persuasion payoff. If $\sigma$
is partitional, our result is immediate. Suppose to the contrary that
in some $\omega$, both $m$ and $m'$ are sent with positive probability.
Then, $m$ and $m'$ must induce the same action: if $m$ induces
some $a$ and $m'$ induces a distinct $a'$, the fact that $u_{S}\left(a,\omega\right)\neq u_{S}\left(a',\omega\right)$
would mean that $\sigma$ cannot be S-BR. Given that any two messages
sent in $\omega$ induce the same action, we can define $\rho\left(\sigma\left(\omega\right)\right)$
as the action that Receiver takes in state $\omega$ given $\left(\sigma,\rho\right)$.

Now, we can consider an alternative, partitional profile $\left(\hat{\sigma},\hat{\rho}\right)$.
Let $f$ be any injective function from $A$ to $M$. Let $\hat{\sigma}\left(\omega\right)=f\left(\rho\left(\sigma\left(\omega\right)\right)\right)$
and $\hat{\rho}\left(f\left(a\right)\right)=a$. It is immediate that
$\left(\hat{\sigma},\hat{\rho}\right)$ is also a cheap-talk equilibrium
and yields the persuasion payoff.

It is perhaps worth noting that Theorems \ref{Thm:commitmentWTA}
and \ref{Thm:commitmentWTP} jointly imply the following:
\begin{cor}
Generically, if cheap-talk Sender values randomization, then committed
Sender values randomization.
\end{cor}

\section{Application to grading}

For an application of our results, we consider their implications
for grading policies. This application also clarifies a sense in which
``randomization'' in the statement of our results need not be interpreted
literally.

Suppose Sender is a school that assigns grades to its students. We
interpret $M$ as the set of potential grades. Each student is characterized
by a vector of attributes. We say an attribute is \emph{relevant}
if an employer values it or the school's value of placing the student
with an employer depends on it. We interpret $\Omega$ as the set
of all possible configurations of the relevant attributes. We maintain
the assumption that $\Omega$ is finite.

Students also have irrelevant attributes. We denote by $X$ as the
set of all possible configurations of the irrelevant attributes. We
assume that the distribution over $X$ is atomless. The school utilizes
a deterministic \emph{grading scheme} $g:\Omega\times X\rightarrow M$.
We say a grading scheme $g$ is \emph{fair} if $g\left(\omega,x\right)=g\left(\omega,x'\right)$
for every $\omega,x,x'$. Otherwise, the scheme is \emph{unfair}. 

For this application, instead of envisioning a single Receiver, we
assume that each student applies to a distinct employer. Each employer
observes the grade $m\in M$ of its applicant and chooses one of finitely
many actions $a\in A$ (e.g., whether to hire the student and if so
for what position). All employers have the same utility function $u_{R}\left(a,\omega\right)$
that depends on the employer's action and the relevant attributes
of the applicant. (If there were a single employer who observed the
grades of all of the applicants, this would effectively provide Sender
with some commitment power because the distribution of messages would
be directly observable to Receiver.) The school's utility is additive
across its students; for each student, the school's payoff $u_{S}\left(a,\omega\right)$
depends on that student's outcome and that student's relevant attributes. 

Under \emph{discretionary grading}, the school freely chooses a grade
to assign to each student, i.e., the school selects any grading scheme
it wishes. The employer only observes its applicant's grade but not
the grading scheme that was used.

Alternatively, the school could implement a (publicly observable)
\emph{grading policy} that restricts the set of schemes that it can
use. A grading policy could be a restriction to one specific grading
scheme. This would make the situation equivalent to Bayesian persuasion.
This is the case even though the grading scheme is deterministic because,
by conditioning the grade on the irrelevant attributes, the school
can implement any distribution of grades conditional on each $\omega$.\footnote{The formulation of experiments as deterministic functions of an expanded
state space was introduced by \citet{gentzkow2017bayesian} and \citet{green2022two}.
It has been further studied in \citet*{brooks2022information} and
\citet*{brooks2024representing}.}

Another type of grading policy is one where the school commits to
a given distribution of grades (\citealt{lin2024credible}). We refer
to such a policy as a\emph{ mandated curve}. For example, the University
of Chicago Law School mandates a pre-specified share of students that
will receive a given narrow range of numerical grades.

More common is commitment to a \emph{GPA cap}. For example, the University
of Chicago Booth School of Business mandates that the average grade
assigned in a given course must not exceed B+.

We say that \emph{the} \emph{school values commitment }if it strictly
prefers to implement any grading policy (full commitment, mandated
curve, GPA cap, etc.) over discretionary grading. We know that any
policy must yield a payoff that is weakly lower than full commitment
and weakly higher than discretionary grading. Consequently, if any
grading policy yields a strictly higher payoff than discretionary
grades, we know that the persuasion payoff (full commitment) exceeds
the cheap talk payoff (discretionary grades).

We say that \emph{the} \emph{school prefers to grade unfairly} if
its optimal grading scheme is unfair. In other words, if the school
were able to commit to a particular grading scheme, it would select
an unfair one.

Theorem \ref{Thm:commitmentWTA} tells us that, generically, the school
values commitment if and only if it prefers to grade unfairly. Thus,
whenever we observe a school mandating a curve or a GPA cap, we know
that the school's ideal policy is unfair.\footnote{Our analysis views the school (that cares about student placements)
and the professor (who is assigning grades) as a single agent. A distinct
motivation for a grading policy such as a GPA cap, outside of our
Sender-Receiver framework, is an agency conflict between the school
and the professor. For example, the professor may wish to give uniformly
high grades in order to avoid student complaints so the school might
impose a GPA cap to mitigate that temptation (\citealt{frankel2014aligned}).
Moreover, a grading policy could have a distinct benefit of aiding
equilibrium coordination about the meaning of grades; our focus on
Sender-preferred equilibria assumes miscoordination away. Finally,
our analysis takes the distribution of relevant attributes as exogenous.
In practice, grading schemes not only provide information about the
students but also incentivize the students to learn the material (\citealt{boleslavsky2015grading}).}

Note, however, that even if we observe a school mandating a curve,
Theorem \ref{Thm:commitmentWTA} does not imply that the school will
implement an unfair scheme if it can only commit to a mandated curve
(i.e., is unable to fully commit to a particular scheme). In Online
Appendix \ref{subsec:Partial}, we analyze whether partial commitment
being valuable (i.e., mandating a curve yields a strictly higher payoff
than discretionary grades) implies that randomization under partial
commitment is valuable (i.e., among the schemes that yield the mandated
curve, every scheme that is optimal is unfair). Under the assumption
that the school's preferences are supermodular, we establish that
this is indeed the case (Theorem \ref{Thm:Curve}). Whether the conclusion
of this result holds when preferences are not supermodular remains
an open question.

\section{How often is commitment valuable?}

Theorems \ref{Thm:commitmentWTA} and \ref{Thm:commitmentWTP} by
themselves do not shed light on the types of environments where commitment
(and thus randomization) are valuable nor about how common such environments
are. This section aims to fill that gap. 

Say that the environment is \emph{felicitous} if Sender obtains his
ideal payoff under cheap talk.\footnote{We thank Roger Myerson for suggesting this term.}
It is immediate that commitment has no value if the environment is
felicitous. The converse, of course, is not generally true (e.g, $u_{S}=-u_{R}$). 

We first show that, for any $\Omega$, the share of environments that
are felicitous is at least $\frac{1}{\left|A\right|^{\left|A\right|}}$.
Thus, in general, the share of environments such that commitment has
no value must be greater than $\frac{1}{\left|A\right|^{\left|A\right|}}$. 

We then establish that, in the limit as $\Omega$ grows large: ($\ast$)
the share of environments that are felicitous converges to exactly
$\frac{1}{\left|A\right|^{\left|A\right|}}$, and ($\ast\ast$) generically,
commitment has no value \emph{if and only if} the environment is felicitous,
i.e., the non-felicitous environments where commitment has no value
(such as $u_{S}=-u_{R}$) become vanishingly rare. Thus, as $\left|\Omega\right|\rightarrow\infty$,
the share of environments such that commitment has no value converges
exactly to $\frac{1}{\left|A\right|^{\left|A\right|}}$. 
\begin{thm}
\label{Thm:how-often}Fix $A$.
\begin{enumerate}
\item[(i)] For any $\Omega$, the share of environments such that commitment
has no value is at least $\frac{1}{\left|A\right|^{\left|A\right|}}$
.
\item[(ii)] As $\left|\Omega\right|\rightarrow\infty$, 
\begin{enumerate}
\item generically, commitment has no value if and only if Sender obtains
his ideal payoff under cheap talk.
\item the share of environments such that commitment has no value converges
to $\frac{1}{\left|A\right|^{\left|A\right|}}$.
\end{enumerate}
\end{enumerate}
\end{thm}
Here we provide a sketch of the proof. For each $a\in A$, let $\Omega_{a}$
be the set of states where $a$ is Sender's ideal action. Generically,
if $a\neq a'$, $\Omega_{a}$ and $\Omega_{a'}$ do not intersect;
consider this case. Then, it is easy to see that the environment is
felicitous if and only if for each $a,a'\in A$, we have\footnote{This condition also appears in \citet*{antic2022subversive} and \citet{aybas2024cheap}.
In \citet*{antic2022subversive}, it is a necessary condition for
the possibility of subversive conversations: without it, a third-party
with veto power would prevent a committee from implementing a project
solely based on the information that the committee wants to do so.
\citet{aybas2024cheap} consider preferences of the form $u_{R}\left(a,\omega\left(\cdot\right)\right)=\omega\left(a\right)^{2}$
and $u_{S}\left(a,\omega\left(\cdot\right)\right)=\left(\omega\left(a\right)-b\right)^{2}$
for some $b>0$ where $\omega:A\rightarrow\mathbb{R}$ is the realized
path of a Brownian motion. They identify features of $b$ and $A$
such that the condition holds.}
\begin{equation}
\sum_{\omega\in\Omega_{a}}\mu_{0}\left(\omega\right)\left(u_{R}\left(a,\omega\right)-u_{R}\left(a',\omega\right)\right)\geq0.\label{eq:obedient}
\end{equation}
Now, for any $\Omega_{a}$ that is not empty, the share of Receiver's
preferences on $A\times\Omega_{a}$ such that inequality $\left(\ref{eq:obedient}\right)$
is satisfied is $\frac{1}{\left|A\right|}$. Thus if all of $\Omega_{a}$'s
are non-empty, the share of environments that are felicitous is $\left(\frac{1}{\left|A\right|}\right)^{\left|A\right|}$,
or $\frac{1}{\left|A\right|^{\left|A\right|}}$.

If an $\Omega_{a}$ is empty, inequality $\left(\ref{eq:obedient}\right)$
is satisfied vacuously for that $a$. Thus, the share of felicitous
environment is weakly greater than $\frac{1}{\left|A\right|^{\left|A\right|}}$.
Since commitment has no value in felicitous environments, we conclude
that commitment has no value in at least $\frac{1}{\left|A\right|^{\left|A\right|}}$
share of environments.

We establish part (ii) of the theorem by showing that as $\left|\Omega\right|$
grows large: ($\ast$) the share of preference such that an $\Omega_{a}$
is empty converges to zero so the share of environments that are felicitous
converges to $\frac{1}{\left|A\right|^{\left|A\right|}}$, and ($\ast\ast$)
the share of environments such that commitment has no value converges
to the share of environments that are felicitous.

Part ($\ast$) is easy to see. For any $a\in A$, as $\Omega$ grows
large, the share of preferences such that there is \emph{no} state
where $a$ is Sender's ideal action converges to zero.

To establish part ($\ast\ast$), say that an environment is \emph{jointly-inclusive
}if for every action $a$, there is some state $\omega$ such that
$a$ is the ideal action for both Sender and Receiver in $\omega$.
Analogously to part ($\ast$), it is easy to see that as $\Omega$
grows large, the share of environments that are jointly-inclusive
converges to $1$. To complete the proof of part ($\ast\ast$), we
argue that, generically, if the environment is jointly-inclusive and
commitment has no value, then the environment must be felicitous.
First, because commitment has no value, Theorem \ref{Thm:commitmentWTP}
implies that there is a partitional profile $\left(\sigma,\rho\right)$
that is a cheap-talk equilibrium and yields the persuasion payoff.
Next, we note that every action $a\in A$ must be induced by $\left(\sigma,\rho\right)$:
there is a state $\omega$ where $a$ is both Sender's and Receiver's
ideal action, so if $a$ were never taken, the committed Sender could
profitably deviate by sometimes\footnote{Sender could reveal $\omega$ with some probability $\varepsilon$;
Receiver's response to all other messages would remain unchanged if
$\varepsilon$ is sufficiently small.} revealing $\omega$ and inducing $a$, thus contradicting the fact
that $\left(\sigma,\rho\right)$ yields the persuasion payoff. This
in turn implies that, for every $\omega$, $\rho\left(\sigma\left(\omega\right)\right)$
must be Sender's ideal action in $\omega$. (If Sender strictly preferred
some other $a'$ in $\omega$, $\left(\sigma,\rho\right)$ could not
be S-BR as the cheap-talk Sender would profitably deviate and set
$\sigma$$\left(\omega\right)$ to be whatever message induces $a'$;
since all actions are induced by $\left(\sigma,\rho\right)$, there
must be such a message.) Thus, $\left(\sigma,\rho\right)$ is a partitional
profile that is R-BR and induces Receiver to take Sender's ideal action
in every state. But this means that every message sent under $\sigma$
fully reveals what action is ideal for Sender, and Receiver complies
and takes that action. Hence, the environment is felicitous.

We conclude this section with a few comments.

First, whether commitment has value in a given environment $\left(u_{S},u_{R}\right)$
depends on the prior $\mu_{0}$. Yet, Theorem \ref{Thm:how-often}
remarkably holds for any (interior) prior.

Second, a corollary of Theorem \ref{Thm:how-often}, part 2b, is that
in the limit $\lim_{\left|A\right|\rightarrow\infty}\lim_{\left|\Omega\right|\rightarrow\infty}$,
commitment generically has value. One might be tempted to summarize
this corollary as ``if both $A$ and $\Omega$ are large, then commitment
is almost certainly valuable.'' An issue with this summary, however,
is that, at least given our current proof method, the conclusion relies
on the order of limits. Whether commitment is generically valuable
when the order of the limits is reversed (i.e., $\lim_{\left|\Omega\right|\rightarrow\infty}\lim_{\left|A\right|\rightarrow\infty})$
or we consider the joint limit (i.e., $\lim_{\substack{|\Omega|\to\infty\\
|A|\to\infty
}
}$) remain open questions.

Third, the felicity condition has a flavor of alignment of Sender
and Receiver's preferences. However, it does not preclude the possibility
that Receiver is much worse off than she would be if Sender and Receiver's
preferences were fully aligned. For instance, consider the prosecutor-judge
example and suppose that the prior is $0.7$ rather than $0.3$; then,
the environment is felicitous but Receiver obtains no information.

\pagebreak{}

\bibliographystyle{plainnat}
\bibliography{randomization}

\pagebreak{}

\appendix

\section{Appendix }

\subsection{Notation and terminology}

Let $A=\left\{ a_{1},...,a_{\left|A\right|}\right\} $. Let $\Omega=\left\{ \omega_{1},...,\omega_{\left|\Omega\right|}\right\} $. 

Given a messaging strategy $\sigma$, let $M_{\sigma}=\{m\in M|\sigma(m|\omega)>0\text{ for some }\omega\}$
be the set of messages that are sent with positive probability under
$\sigma$. For any $\omega$, if $\sigma(\cdot|\omega)$ is degenerate
(i.e., there exists a message $m$ such that $\sigma(m|\omega)=1$),
let $\sigma(\omega)$ denote the message that is sent in state $\omega$.
Similarly, if $\rho(\cdot|m)$ is degenerate, let $\rho\left(m\right)$
denote the action taken following message $m$.

Say that $\rho$ is \textit{pure} if $\rho(\cdot|m)$ is degenerate
for all $m\in M$. Given a profile $\left(\sigma,\rho\right)$, say
$\rho$ is \emph{pure-on-path} if $\rho\left(\cdot|m\right)$ is degenerate
for all $m\in M_{\sigma}$.

Let $I=\left[\begin{array}{c}
e_{1}\\
...\\
e_{\left|\Omega\right|}
\end{array}\right]$ be the identity matrix of size $\left|\Omega\right|$, with $e_{i}$
being the row vector with all entries equal to $0$ except for the
$i^{th}$ entry equal to $1$. Denote a vector all of whose elements
are equal to $r$ by $\boldsymbol{r}$. Denote the $j^{th}$ element
of $\mu$ by $\left[\mu\right]_{j}$. 

\subsection{Generic environments for the proofs}

We now introduce two generic sets of environments that play important
roles in the proofs. 

\subsubsection{Partitional-unique-response environments}

An environment $\left(u_{S},u_{R}\right)$ satisfies \emph{partitional-unique-response}
if for every non-empty $\hat{\Omega}\subseteq\Omega$, 
\[
\arg\max_{a\in A}\sum_{\omega\in\hat{\Omega}}\mu_{0}(\omega)u_{R}(a,\omega)
\]
is a singleton.

Note that whether an environment satisfies partitional-unique-response
does not depend on Sender's preferences. The partitional-unique-response
property requires that, at the finitely many beliefs that can be induced
by a partitional experiment, Receiver has a unique best response.
\begin{lem}
\label{lemma:R-unique-response-genericity} The set of partitional-unique-response
environments is generic.
\end{lem}
\begin{proof}
Let $\mathcal{T}$ denote the set of triplets $\left(\hat{\Omega},a_{i},a_{j}\right)$
such that $\emptyset\neq\hat{\Omega}\subseteq\Omega$, $a_{i},a_{j}\in A$,
and $a_{i}\neq a_{j}$. For any $\left(\hat{\Omega},a_{i},a_{j}\right)\in\mathcal{T}$,
let $Q(\hat{\Omega},a_{i},a_{j})$ denote the set of $u_{R}$ such
that 
\begin{equation}
\sum_{\omega\in\hat{\Omega}}\,\mu_{0}(\omega)u_{R}(a_{i},\omega)=\sum_{\omega\in\hat{\Omega}}\,\mu_{0}(\omega)u_{R}(a_{j},\omega).\label{eq:not-unique-BR}
\end{equation}

We will show that $\cup_{\left(\hat{\Omega},a_{i},a_{j}\right)\in\mathcal{T}}Q(\hat{\Omega},a_{i},a_{j})$
has measure zero in $[0,1]^{|A||\Omega|}$. Since $A$ and $\Omega$
are finite, it suffices to show that for any $\left(\hat{\Omega},a_{i},a_{j}\right)\in\mathcal{T}$,
$Q(\hat{\Omega},a_{i},a_{j})$ has measure zero.

Fix any $\left(\hat{\Omega},a_{i},a_{j}\right)\in\mathcal{T}$. Note
that $Q(\hat{\Omega},a_{i},a_{j})$ can be written as
\begin{equation}
\left\{ u_{R}\in\left[0,1\right]^{|A||\Omega|}|\sum_{\omega,a}u_{R}(a,\omega)\eta(a,\omega)=0\right\} \label{eq:generic-2}
\end{equation}
where 
\[
\eta(a,\omega)=\begin{cases}
\mu_{0}(\omega) & \text{ if }a=a_{i},\omega\in\hat{\Omega}\\
-\mu_{0}(\omega) & \text{ if }a=a_{j},\omega\in\hat{\Omega}\\
0 & \text{ otherwise.}
\end{cases}
\]
Hence, $Q(\hat{\Omega},a_{i},a_{j})$ is a subset of a hyperplane
in $\mathbb{R}^{|A||\Omega|}$, and thus has measure zero.\footnote{The set of partitional-unique-response environment is also generic
in the topological sense. It is easy to see that $Q(\hat{\Omega},a_{i},a_{j})$
is closed. Since $\cup_{\left(\hat{\Omega},a_{i},a_{j}\right)\in\mathcal{T}}Q(\hat{\Omega},a_{i},a_{j})$
is therefore closed, its complement is open. Since $\cup_{\left(\hat{\Omega},a_{i},a_{j}\right)\in\mathcal{T}}Q(\hat{\Omega},a_{i},a_{j})$
has measure zero, its complement is dense. Thus, the set of partitional-unique-response
environments, which is a superset of the complement of $\cup_{\left(\hat{\Omega},a_{i},a_{j}\right)\in\mathcal{T}}Q(\hat{\Omega},a_{i},a_{j})$,
contains an open, dense set. \label{fn:puq-topology}}
\end{proof}

\subsubsection{Scant-indifferences environments }

For each $a_{i}\in A$, let $\bdu_{S}(a_{i})=u_{S}(a_{i},\cdot)\in\mathbb{R}^{\left|\Omega\right|}$
and $\bdu_{R}(a_{i})=u_{R}(a_{i},\cdot)\in\mathbb{R}^{|\Omega|}$
denote the payoff vectors across states.

For each $a_{i},$ define the \emph{expanded-indifference matrix}
$T^{i}$ as follows. Let $T_{S}^{i}$ be the matrix with $\left|A\right|-1$
rows and $\left|\Omega\right|$ columns, with each row associated
with $j\neq i$ and equal to $\bdu_{S}(a_{j})-\bdu_{S}(a_{i})$. Let
$T_{R}^{i}$ be the matrix with $\left|A\right|-1$ rows and $\left|\Omega\right|$
columns, with each row associated with $j\neq i$ and equal to $\bdu_{R}(a_{j})-\bdu_{R}(a_{i})$.
Then, let
\[
T^{i}=\begin{bmatrix}T_{S}^{i}\\
T_{R}^{i}\\
I
\end{bmatrix}.
\]

Given any matrix $T$, a \emph{row-submatrix} of $T$ is a matrix
formed by removing some of the rows of $T$.

We say that an environment satisfies \emph{scant-indifferences} if
for each $a_{i}\in A$, every row-submatrix of the expanded-indifference
matrix $T^{i}$ is full rank.

We anticipate that the reader might find this definition mysterious,
so we now try to provide some intuition by connecting this definition
to the proof sketch we gave in the body of the paper for Theorem \ref{Thm:commitmentWTP}
in the case with two actions and three states. 

Recall, that in Figure \ref{fig:Indifference-optimality}, the argument
behind Lemma \ref{lem:Receiver-not-mix} relied on two facts that
must hold generically. First, the border between $R_{1}$ and $R_{2}$
is distinct from the border between $S_{1}$ and $S_{2}$ and thus
the two borders have at most one intersection, $\mu_{m}$. Second,
generically $\mu_{m}$ (if it exists) is an interior belief. Moreover,
the argument behind Lemma \ref{lem:Sender-not-mix} relied on the
fact that, generically, for any $\omega$ and $a_{i}\neq a_{j}$,
$u_{S}(a_{i},\omega)\neq u_{S}(a_{j},\omega)$. 

We now illustrate why these three facts hold in any scant-indifferences
environment. With only two actions, we can look at $T^{1}$ only,
since the argument for $T^{2}$ is identical. We have
\[
T^{1}=\begin{bmatrix}\begin{array}{ccc}
\overset{\Delta}{u}_{S}\left(\omega_{1}\right) & \overset{\Delta}{u}_{S}\left(\omega_{2}\right) & \overset{\Delta}{u}_{S}\left(\omega_{3}\right)\\
\overset{\Delta}{u}_{R}\left(\omega_{1}\right) & \overset{\Delta}{u}_{R}\left(\omega_{2}\right) & \overset{\Delta}{u}_{R}\left(\omega_{3}\right)\\
1 & 0 & 0\\
0 & 1 & 0\\
0 & 0 & 1
\end{array}\end{bmatrix}
\]
where $\overset{\Delta}{u}_{S}\left(\omega_{i}\right)=u_{S}\left(a_{2},\omega_{i}\right)-u_{S}\left(a_{1},\omega_{i}\right)$
and analogously for $\overset{\Delta}{u}_{R}$.

First, consider the row-submatrix 
\[
\overset{\Delta}{T}=\begin{bmatrix}\overset{\Delta}{u}_{S}\left(\omega_{1}\right) & \overset{\Delta}{u}_{S}\left(\omega_{2}\right) & \overset{\Delta}{u}_{S}\left(\omega_{3}\right)\\
\overset{\Delta}{u}_{R}\left(\omega_{1}\right) & \overset{\Delta}{u}_{R}\left(\omega_{2}\right) & \overset{\Delta}{u}_{R}\left(\omega_{3}\right)
\end{bmatrix}.
\]
Note that both Sender and Receiver are indifferent between the two
actions at a belief $\mu$ if and only if $\overset{\Delta}{T}\mu=0$.
Thus, requiring that $\overset{\Delta}{T}$ be full-rank is equivalent
to requiring that the border between $R_{1}$ and $R_{2}$ not be
parallel to the border between $S_{1}$ and $S_{2}$. A fortiori,
the environment satisfying scant-indifferences implies that the two
borders do not coincide.

Second, consider the row-submatrix 
\[
\begin{bmatrix}\overset{\Delta}{u}_{S}\left(\omega_{1}\right) & \overset{\Delta}{u}_{S}\left(\omega_{2}\right) & \overset{\Delta}{u}_{S}\left(\omega_{3}\right)\\
\overset{\Delta}{u}_{R}\left(\omega_{1}\right) & \overset{\Delta}{u}_{R}\left(\omega_{2}\right) & \overset{\Delta}{u}_{R}\left(\omega_{3}\right)\\
1 & 0 & 0
\end{bmatrix}.
\]
Requiring that this matrix be full-rank yields that $\mu_{m}$ puts
strictly positive probability on $\omega_{1}$. Considering the row-submatrices
that alternatively include the other two rows of the identity matrix
yields that $\mu_{m}$ puts strictly positive probability on $\omega_{2}$
and $\omega_{3}$.

Finally, suppose that in, say state $\omega_{1}$, $\overset{\Delta}{u}_{S}\left(\omega_{1}\right)=0$.
Consider the row-submatrix
\[
\begin{bmatrix}0 & \overset{\Delta}{u}_{S}\left(\omega_{2}\right) & \overset{\Delta}{u}_{S}\left(\omega_{3}\right)\\
0 & 1 & 0\\
0 & 0 & 1
\end{bmatrix}.
\]
Clearly, this matrix is not full-rank, so scant-indifferences rules
out the possibility that $u_{S}(a_{1},\omega_{1})=u_{S}(a_{2},\omega_{1})$.

Having motivated the definition of scant-indifferences environments,
we now establish that the set of such environments is generic.
\begin{lem}
\label{lemma:full-rank-genericity}The set of scant-indifferences
environments is generic. 
\end{lem}
\begin{proof}
First, observe that given any expanded-indifference matrix $T^{i}$,
if every square row-submatrix of $T^{i}$ is full-rank, then every
row-submatrix of $T^{i}$ is full-rank. To see why, suppose every
square row-submatrix of $T^{i}$ is full-rank. Now, consider an arbitrary
row-submatrix $\hat{T}$ of $T^{i}$. If $\hat{T}$ is square, it
obviously has full-rank. Suppose that $\hat{T}$ has more than $\left|\Omega\right|$
rows. In that case, every square row-submatrix of $\hat{T}$ is also
a square row-submatrix of $T^{i}$. This row-submatrix has rank $\left|\Omega\right|$.
Therefore, $\hat{T}$ has rank $\left|\Omega\right|$ and is thus
full-rank. Finally, suppose that $\hat{T}$ has fewer than $\left|\Omega\right|$
rows. We know that $\hat{T}$ is a row-submatrix of some square row-submatrix
$\tilde{T}$ of $T^{i}$. We know $\tilde{T}$ has full-rank so all
of its rows are linearly independent. Consequently, the subset of
its rows that constitute $\hat{T}$ is also linearly independent. 

Hence, we can consider only square row-submatrices of $T^{i}$. Recall
that a square matrix is full-rank if and only if its determinant is
non-zero. Thus, it will suffice to show that for a full Lebesgue measure
set of $\left(u_{S},u_{R}\right)$, the determinant of every square
row-submatrix of each expanded-indifference matrix is non-zero. Consider
any square row-submatrix $\hat{T}$ of any expanded-indifference matrix.
The determinant of $\hat{T}$ is a non-zero polynomial function of
$(u_{S},u_{R})$.\footnote{To establish that this polynomial function is not identically zero,
we need to show that there are values of $(u_{S},u_{R})$ such that
the determinant of $\hat{T}$ is not zero. Recall that rearranging
the rows of a matrix does not change whether the determinant is zero.
Let $\tilde{T}$ be some matrix constructed by rearranging the rows
of $\hat{T}$ so that if $\tilde{T}$ contains any row equal to $e_{k}$,
that row is the $k^{th}$ row of $\tilde{T}$. Any other row $\ell$
of $\tilde{T}$ must belong either to $T_{S}^{i}$ or $T_{R}^{i}$.
Construct $\left(u_{S},u_{R}\right)$ as follows. If $\bdu_{S}(a_{j})-\bdu_{S}(a_{i})$
is not in $\hat{T}$ (and thus not in $\tilde{T}$), set $\bdu_{S}(a_{j})$
in an arbitrary way. If $\bdu_{S}(a_{j})-\bdu_{S}(a_{i})$ is the
$\ell^{th}$ row of $\tilde{T}$, set $\bdu_{S}(a_{j})$ as the row
vector with all entries equal to $0$ except for the $\ell^{th}$
entry equal to $1$. Similarly, if $\bdu_{R}(a_{j})-\bdu_{R}(a_{i})$
is not in $\hat{T}$ (and thus not in $\tilde{T}$), set $\bdu_{R}(a_{j})$
in an arbitrary way, and if $\bdu_{R}(a_{j})-\bdu_{R}(a_{i})$ is
the $\ell^{th}$ row of $\tilde{T}$, set $\bdu_{R}(a_{j})$ as the
row vector with all entries equal to $0$ except for the $\ell^{th}$
entry equal to $1$. Finally, set $\bdu_{S}(a_{i})=\bdu_{R}(a_{i})=\boldsymbol{0}$.
With $(u_{S},u_{R})$ defined this way, $\tilde{T}$ is an identity
matrix with determinant 1. Therefore, the determinant of $\hat{T}$
is a non-zero function of $(u_{S},u_{R})$.} The zero set of any non-zero polynomial function has Lebesgue measure
zero, so the set of $\left(u_{S},u_{R}\right)$ for which $\hat{T}$
does not have full rank is a measure-zero set. Since there are only
finitely many square row-submatrices of expanded-indifference matrices,
the set of scant-indifferences environments is generic.\footnote{The zero set of any non-zero polynomial function is closed, so the
set of scant indifferences environments is generic in the topological
sense as well.\label{fn:scant-topology}}
\end{proof}
As we noted above (for the three state, two action case), in scant-indifferences
environments, there is no state in which Sender is indifferent between
two distinct actions.
\begin{lem}
\label{lem:no-state-difference-between-two-actions}In any scant-indifferences
environment, for any $\omega$ and $a_{i}\neq a_{j}$, $u_{S}(a_{i},\omega)\neq u_{S}(a_{j},\omega)$.
\end{lem}
\begin{proof}
Suppose, toward a contradiction, that there exist some $\omega$,
$a_{i}$, and $a_{j}$ such that $u_{S}(a_{i},\omega)=u_{S}(a_{j},\omega)$.
Without loss, suppose this holds for $\omega_{1}$. Then, the vector
$\bdu_{S}(a_{i})-\bdu_{S}(a_{j})$ has zero as its first element.
Now consider the $|\Omega|\times|\Omega|$ row sub-matrix of $T^{j}$
\[
\begin{bmatrix}\bdu_{S}(a_{i})-\bdu_{S}(a_{j})\\
e_{2}\\
...\\
e_{|\Omega|}
\end{bmatrix}.
\]
This matrix is not full-rank because the first row can be expressed
as a linear combination of the other rows.
\end{proof}

\subsection{\protect\label{subsec:purity}Key Proposition}

In this section we establish a key proposition that will be useful
for proofs of Theorems \ref{Thm:commitmentWTA}, \ref{Thm:commitmentWTP},
and \ref{Thm:how-often}.
\begin{prop}
\label{Lemma:purity}In a scant-indifferences environment, if commitment
has no value, then there is a partitional $\hat{\sigma}$ and a pure
strategy $\hat{\rho}$ such that: (i) $(\hat{\sigma},\hat{\rho})$
is a cheap-talk equilibrium and yields the persuasion payoff, and
(ii) $\left|M_{\hat{\sigma}}\right|\leq\left|A\right|$.
\end{prop}
To establish the Proposition, we first show that if a cheap-talk equilibrium
(in fact any R-BR profile) yields the persuasion payoff, then Receiver
must not randomize on path in that equilibrium (Lemma \ref{lem:Receiver-not-mix}).
Second, we show that if Receiver does not randomize on path, Sender
also need not randomize (Lemma \ref{lem:Sender-not-mix}). 
\begin{lem}
\label{lem:Receiver-not-mix}In a scant-indifferences environment,
if $\left(\sigma,\rho\right)$ is R-BR and yields the persuasion payoff,
then $\rho$ must be pure-on-path.
\end{lem}
\begin{proof}
Suppose by contradiction that the environment satisfies scant-indifferences,
profile $\left(\sigma,\rho\right)$ is R-BR and yields the persuasion
payoff, yet there exists a message $m\in M_{\sigma}$ such that $|\Supp(\rho(\cdot|m))|=k>1$.

We first note that both Sender and Receiver must be indifferent among
all the actions in $\Supp(\rho(\cdot|m))$ given $\mu_{m}$, the belief
induced by message $m$. In other words, for all $a_{i},a_{j}\in\Supp(\rho(\cdot|m))$,
\begin{align}
\sum_{\omega}\mu_{m}(\omega)u_{R}(a_{i},\omega)=\sum_{\omega}\mu_{m}(\omega)u_{R}(a_{j},\omega),\label{eq:R-inddifference}\\
\sum_{\omega}\mu_{m}(\omega)u_{S}(a_{i},\omega)=\sum_{\omega}\mu_{m}(\omega)u_{S}(a_{j},\omega).\label{eq:S-inddifference}
\end{align}
Equation \eqref{eq:R-inddifference} follows immediately from R-BR.
Equation \eqref{eq:S-inddifference} follows from the fact that $(\sigma,\rho)$
yields the persuasion payoff: if say $\sum_{\omega}\mu_{m}(\omega)u_{S}(a_{i},\omega)>\sum_{\omega}\mu_{m}(\omega)u_{S}(a_{j},\omega)$,
an alternative strategy profile where Receiver breaks ties in favor
of Sender would still satisfy R-BR while strictly improving Sender's
payoff.

For each belief $\mu\in\Delta\Omega$, let $A_{R}^{*}(\mu)$ denote
the set of Receiver-optimal actions under belief $\mu$; that is,
$A_{R}^{*}(\mu)=\arg\max_{a\in A}\bdu_{R}(a)\cdot\mu$. Clearly, $\Supp(\rho(\cdot|m))\subseteq A_{R}^{*}(\mu_{m})$,
meaning that $A_{R}^{*}(\mu_{m})$ contains the $k$ actions in the
support of $\rho(\cdot|m)$, but may also contain additional actions
that are not played following $m$. Without loss of generality, let
$\Supp(\rho(\cdot|m))=\{a_{1},...,a_{k}\}$ and $A_{R}^{*}(\mu_{m})=\{a_{1},...,a_{k},a_{k+1},...,a_{k+r}\}$
for some $r\geq0$. Note that for any $i=2,...,k+r,$ $\bdu_{R}(a_{1})\cdot\mu_{m}=\bdu_{R}(a_{i})\cdot\mu_{m}$.

Equation \eqref{eq:S-inddifference} implies that for any $i=2,...,k,$
$\bdu_{S}(a_{1})\cdot\mu_{m}=\bdu_{S}(a_{i})\cdot\mu_{m}$. Combining
both Sender's and Receiver's indifference conditions, we have
\begin{equation}
\begin{bmatrix}\bdu_{S}(a_{2})-\bdu_{S}(a_{1})\\
...\\
\bdu_{S}(a_{k})-\bdu_{S}(a_{1})\\
\\\bdu_{R}(a_{2})-\bdu_{R}(a_{1})\\
...\\
\bdu_{R}(a_{k+r})-\bdu_{R}(a_{1})
\end{bmatrix}\mu_{m}=\boldsymbol{0}.\label{eq:indifference}
\end{equation}

Let $\hat{\Omega}=\{\omega|\mu_{m}(\omega)=0\}$, the (potentially
empty) set of states that are not in the support of $\mu_{m}$. Without
loss, suppose that $\hat{\Omega}=\left\{ \omega_{1},...\omega_{\ell}\right\} $
where $\ell\geq0$. If $\ell>0$ (i.e., $\hat{\Omega}\neq\emptyset$),
then we have 
\begin{equation}
\begin{bmatrix}e_{1}\\
...\\
e_{\ell}
\end{bmatrix}\mu_{m}=\boldsymbol{0}.\label{eq:facet}
\end{equation}

Let $\hat{T}_{S}=\begin{bmatrix}\bdu_{S}(a_{2})-\bdu_{S}(a_{1})\\
...\\
\bdu_{S}(a_{k})-\bdu_{S}(a_{1})
\end{bmatrix}$, $\hat{T}_{R}=\begin{bmatrix}\bdu_{R}(a_{2})-\bdu_{R}(a_{1})\\
...\\
\bdu_{R}(a_{k+r})-\bdu_{R}(a_{1})
\end{bmatrix}$, $\hat{E}=\begin{bmatrix}e_{1}\\
...\\
e_{\ell}
\end{bmatrix}$, and $\hat{T}=\begin{bmatrix}\hat{T}_{S}\\
\hat{T}_{R}\\
\hat{E}
\end{bmatrix}$. Note that $\hat{T}$ is a row-submatrix of the expanded-indifference
matrix $T^{1}$.

Combining \eqref{eq:indifference} and \eqref{eq:facet}, we know
$\hat{T}\mu_{m}=\boldsymbol{0}$. Moreover, since $\mu_{m}\in\Delta\Omega$,
we know $\boldsymbol{1}\mu_{m}=1$. 

Next we make two observations: (i) $rank(\hat{T})<|\Omega|$, otherwise
the unique solution to $\text{\ensuremath{\hat{T}}}\mu=\boldsymbol{0}$
is $\mu=\boldsymbol{0}$. Since we are in a scant-indifferences environment,
this means that $\hat{T}$ has full row rank; (ii) vector $\boldsymbol{1}$
can not be represented as a linear combination of rows of $\hat{T}$.
To see why, assume toward contradiction that there exists a row vector
$\lambda\in\mathbb{R}^{2k+r+\ell-2}$ such that $\lambda\hat{T}=\boldsymbol{1}$.
This would lead to a contradiction that $1=\boldsymbol{1}\mu_{m}=\lambda\hat{T}\mu_{m}=\lambda\boldsymbol{0}=0$. 

Observations (i) and (ii) together imply that the matrix $\begin{bmatrix}\hat{T}\\
\boldsymbol{1}
\end{bmatrix}$ has full row rank. Consequently, we know $rank\left(\left[\begin{array}{c}
\text{\ensuremath{\hat{T}}}\\
\boldsymbol{1}
\end{array}\right]\right)>rank\begin{pmatrix}\begin{bmatrix}\hat{T}_{R}\\
\hat{E}\\
\boldsymbol{1}
\end{bmatrix}\end{pmatrix}.$

Now, we claim that there exists $x\in\mathbb{R}^{n}$ such that 
\begin{equation}
\begin{bmatrix}\hat{T}_{R}\\
\hat{E}\\
\boldsymbol{1}
\end{bmatrix}x=0\label{eq:R-kernel}
\end{equation}
and 
\begin{equation}
\hat{T}_{S}\,x\neq0.\label{eq:S-kernel}
\end{equation}
To see this, suppose by contradiction that for any $x$ that solves
\eqref{eq:R-kernel}, we have $\hat{T}_{S}\,x=0$. This would imply
that the set of solutions to \eqref{eq:R-kernel} and the set of solutions
to
\begin{equation}
\begin{bmatrix}\hat{T}\\
\boldsymbol{1}
\end{bmatrix}x=0\label{eq:both-kernel}
\end{equation}
coincide. By the Rank-Nullity Theorem, however, the subspace defined
by \eqref{eq:both-kernel} has dimension $|\Omega|-rank\left(\begin{bmatrix}\hat{T}\\
\boldsymbol{1}
\end{bmatrix}\right)$, while the subspace defined by \eqref{eq:R-kernel} has a higher
dimension $|\Omega|-rank\left(\begin{bmatrix}\hat{T}_{R}\\
\hat{E}\\
\boldsymbol{1}
\end{bmatrix}\right)$.

Consider an $x$ that satisfies \eqref{eq:R-kernel} and \eqref{eq:S-kernel}.
Consider two vectors, $\mu_{m}+\varepsilon x$ and $\mu_{m}-\varepsilon x$,
for some $\varepsilon\in\mathbb{R}_{>0}$. First we verify that for
sufficiently small $\varepsilon$, $\mu_{m}\pm\varepsilon x\in\Delta\text{\ensuremath{\Omega}}$.
Since $\boldsymbol{1}x=0$, it follows that $\boldsymbol{1}\left(\mu_{m}\pm\varepsilon x\right)=\boldsymbol{1}\mu_{m}=1$.
For $\omega_{j}\notin\hat{\Omega}$, we have $\left[\mu_{m}\right]_{j}>0$,
so for small enough $\varepsilon$, $\left[\mu_{m}\pm\varepsilon x\right]_{j}\geq0$.
For $\omega_{j}\in\hat{\Omega}$, we know $e_{j}$ is a row of $\hat{E}$,
so $e_{j}x=0$. Consequently, $\left[\mu_{m}\pm\varepsilon x\right]_{j}=e_{j}\left(\mu_{m}\pm\varepsilon x\right)=\left[\mu_{m}\right]_{j}=0.$
Thus, $\mu_{m}\pm\varepsilon x\in\Delta\text{\ensuremath{\Omega}}$.

Observe that $A_{R}^{*}(\mu_{m})=A_{R}^{*}(\mu_{m}\pm\varepsilon x)$.
First, for any $a\notin A_{R}^{*}(\mu_{m})$, if $\varepsilon$ is
sufficiently small, $a\notin A_{R}^{*}(\mu_{m}\pm\varepsilon x)$.
Therefore, $A_{R}^{*}(\mu_{m}\pm\varepsilon x)\subseteq A_{R}^{*}(\mu_{m})$.
But, $\hat{T}_{R}\,x=0$ implies that $\left(\mu_{m}\pm\varepsilon x\right)\cdot\bdu_{R}(a)$
is constant across $a\in A_{R}^{*}(\mu_{m})$, so $A_{R}^{*}(\mu_{m}\pm\varepsilon x)=A_{R}^{*}(\mu_{m})$.

Consider an alternative messaging strategy $\hat{\sigma}$ that is
identical to $\sigma$, except that the message $m$ is split into
two new messages, $m^{+}$ and $m^{-}$, which induce the beliefs
$\mu_{m}+\varepsilon x$ and $\mu_{m}-\varepsilon x$, respectively.\footnote{It is possible for $M_{\sigma}=M$, but we can consider an alternative
strategy that induces the same outcome as $\sigma$ and uses only
$\left|A\right|$ messages. We can also let $m$ play the role of
$m^{+}$ or $m^{-}$, so our assumption that $\left|M\right|\geq\left|A\right|+1$
suffices. } We consider $\hat{\rho}$ that agrees with $\rho$ on messages other
than $\left\{ m,m^{+},m^{-}\right\} $ and leads Receiver to break
indifferences in Sender's favor following $m^{+}$ and $m^{-}$. We
will show that $\left(\hat{\sigma},\hat{\rho}\right)$ yields a strictly
higher payoff to Sender, thus contradicting the assumption that $\left(\sigma,\rho\right)$
yields the persuasion payoff.

Since $\hat{T}_{S}\,x\neq0$, we know there is an $a_{i}\in\{a_{2},...,a_{k}\}$
such that $x\cdot\left(\bdu_{S}(a_{i})-\bdu_{S}(a_{1})\right)\neq0$.

Because $a_{1}\in A_{R}^{*}(\mu_{m}\pm\varepsilon x)=A_{R}^{*}(\mu_{m})$,
we have
\[
\max_{a\in A^{*}(\mu_{m})}\left(\mu_{m}+\varepsilon x\right)\cdot\left(\bdu_{S}(a)-\bdu_{S}(a_{1})\right)\geq0
\]
and 
\[
\max_{a\in A^{*}(\mu_{m})}\left(\mu_{m}-\varepsilon x\right)\cdot\left(\bdu_{S}(a)-\bdu_{S}(a_{1})\right)\geq0.
\]
We now establish that at least one of these inequalities has to be
strict. Suppose toward contradiction that both hold with equality.
The first equality implies that for all $a_{i}\in\{a_{2},...,a_{k}\}$,
$\left(\mu_{m}+\varepsilon x\right)\cdot\left(\bdu_{S}(a_{i})-\bdu_{S}(a_{1})\right)\leq0$,
which combined with the fact that $\mu_{m}\cdot\bdu_{S}(a_{i})=\mu_{m}\cdot\bdu_{S}(a_{1})$
implies that $x\cdot\left(\bdu_{S}(a_{i})-\bdu_{S}(a_{1})\right)\leq0$
for all $a_{i}\in\{a_{2},...,a_{k}\}$. Similarly, the second equality
implies that $-x\cdot\left(\bdu_{S}(a_{i})-\bdu_{S}(a_{1})\right)\leq0$
for all $a_{i}\in\{a_{2},...,a_{k}\}$. Together, this yields that
$x\cdot\left(\bdu_{S}(a_{i})-\bdu_{S}(a_{1})\right)=0$ for all $a_{i}\in\{a_{2},...,a_{k}\}$,
a contradiction. Hence, one of the inequalities has to be strict.

Consequently, Sender's interim payoff under $\hat{\sigma}$ (in the
event that $m$ is sent under $\sigma$) is 
\begin{align*}
 & \frac{1}{2}\max_{a\in A^{*}(\mu_{m})}\left(\mu_{m}+\varepsilon x\right)\cdot\bdu_{S}(a)+\frac{1}{2}\max_{a\in A^{*}(\mu_{m})}\left(\mu_{m}-\varepsilon x\right)\cdot\bdu_{S}(a)\\
> & \frac{1}{2}\left(\mu_{m}+\varepsilon x\right)\cdot\bdu_{S}(a_{1})+\frac{1}{2}\left(\mu_{m}-\varepsilon x\right)\cdot\bdu_{S}(a_{1})\\
= & \mu_{m}\cdot\bdu_{S}(a_{1})
\end{align*}
Thus, $\left(\hat{\sigma},\hat{\rho}\right)$ yields a strictly higher
payoff to Sender, contradicting the assumption that $\left(\sigma,\rho\right)$
yields the persuasion payoff.
\end{proof}
\begin{lem}
\label{lem:Sender-not-mix}In a scant-indifferences environment, if
a cheap-talk equilibrium $(\sigma,\rho)$ yields the persuasion payoff
and $\rho$ is pure-on-path, then there exists a partitional $\hat{\sigma}$
and a pure strategy $\hat{\rho}$ such that $\left|M_{\hat{\sigma}}\right|\leq\left|A\right|$
and $(\hat{\sigma},\hat{\rho})$ is a cheap-talk equilibrium and yields
the persuasion payoff.
\end{lem}
\begin{proof}
Suppose a cheap-talk equilibrium $(\sigma,\rho)$ yields the persuasion
payoff and $\rho$ is pure-on-path.

First, we show that for any $\omega$ and any $m,m'$ such that $\sigma(m|\omega),\sigma(m'|\omega)>0$,
$\rho(m)=\rho(m')$. The fact that both $m$ and $m'$ are sent in
$\omega$ implies, by S-BR, that $u_{S}(\rho(m),\omega)=u_{S}(\rho(m'),\omega)$.
Moreover, by Lemma \ref{lem:no-state-difference-between-two-actions},
there exist no distinct $a$ and $a'$ such that $u_{S}(a,\omega)=u_{S}(a',\omega)$,
so it must be that $\rho(m)=\rho(m')$.

Let $A^{*}=\{a\in A|a=\rho(m)\text{ for some }m\in M_{\sigma}\}$
be the set of actions that are taken on-path. Without loss, let $A^{*}=\{a_{1},...,a_{k}\}$.
For each $a_{i}$, let $M_{i}=\{m\in M_{\sigma}|\rho(m)=a_{i}\}$
be the set of on-path messages that induce action $a_{i}$, and $\Omega_{i}=\{\omega\in\Omega|\Supp(\sigma(\cdot|\omega))\subseteq M_{i}\}$
be the set of states that induce action $a_{i}$. Note that $\left\{ M_{i}\right\} _{i=1}^{k}$
is a partition of $M_{\sigma}.$ Moreover, it is easy to see that
$\{\Omega_{i}\}_{i=1}^{k}$ is a partition of $\Omega$. First, $\Omega_{i}$
cannot be empty because every $a_{i}\in A^{*}$ is taken on-path.
Second, every $\omega\in\Omega$ belongs to some $\Omega_{i}$ as
only actions in $A^{*}$ are taken on-path; hence, $\cup_{i}\Omega_{i}=\Omega$.
Finally, the fact that for any $\omega$ and any $m,m'$ such that
$\sigma(m|\omega),\sigma(m'|\omega)>0$ we have $\rho(m)=\rho(m')$
implies that if $i\neq j$, $\Omega_{i}$ and $\Omega_{j}$ are disjoint.

Now select one message in each $M_{i}$, and label it as $m_{i}$.
Consider the following alternative strategy profile $(\hat{\sigma},\hat{\rho})$: 
\begin{itemize}
\item $\hat{\sigma}(m_{i}|\omega)=1$ if $\omega\in\Omega_{i}$. 
\item $\hat{\rho}(m_{i})=a_{i}.$ 
\item $\hat{\rho}(m)=a_{1}$ if $m\in M\backslash\{m_{1},....m_{k}\}$. 
\end{itemize}
Note that $\hat{\sigma}$ is well defined because $\{\Omega_{i}\}_{i=1}^{k}$
is a partition of $\Omega$. By construction, $\hat{\sigma}$ is partitional,
$\left|M_{\hat{\sigma}}\right|\leq\left|A\right|,$ and $\hat{\rho}$
is a pure strategy. Moreover, under both $(\sigma,\rho)$ and $(\hat{\sigma},\hat{\rho})$,
every state in $\Omega_{i}$ induces action $a_{i}$ with probability
1. Thus, the two strategy profiles induce the same distribution over
states and actions, so $(\hat{\sigma},\hat{\rho})$ also yields the
persuasion payoff. It remains to show that $(\hat{\sigma},\hat{\rho})$
is a cheap-talk equilibrium.

Note that S-BR of $(\sigma,\rho)$ implies that for any $\omega$
and $m\in\Supp(\sigma(\cdot|\omega))$, we have
\[
u_{S}(\rho(m),\omega)\geq u_{S}(\rho(m'),\omega)\text{ for all }m'\in M_{\sigma}.
\]
Therefore, for any $\omega\in\Omega_{i}$, $u_{S}(a_{i},\omega)\geq u_{S}(a_{j},\omega)$
for all $a_{j}\in A^{*}$. This implies that $u_{S}(\hat{\rho}(\hat{\sigma}(\omega)),\omega)\geq u_{S}(\hat{\rho}(m'),\omega)$
for all $m'\in M$. Hence, $(\hat{\sigma},\hat{\rho})$ satisfies
S-BR.

The fact that $(\sigma,\rho)$ is R-BR implies that for all $m\in M_{\sigma}$,
\[
\sum_{\omega\in\Omega}\mu_{0}(\omega)\sigma(m|\omega)u_{R}(\rho(m),\omega)\geq\sum_{\omega\in\Omega}\mu_{0}(\omega)\sigma(m|\omega)u_{R}(a',\omega)\quad\text{for all }a'\in A.
\]
For any $i\in\left\{ 1,...,k\right\} $, we can sum the inequality
above over $m\in M_{i}$. Since for $m\in M_{i}$ we have $\rho\left(m\right)=a_{i}$,
this yields
\[
\sum_{\omega\in\Omega}\mu_{0}(\omega)\sum_{m\in M_{i}}\sigma(m|\omega)u_{R}(a_{i},\omega)\geq\sum_{\omega\in\Omega}\mu_{0}(\omega)\sum_{m\in M_{i}}\sigma(m|\omega)u_{R}(a',\omega)\quad\text{for all }a'\in A.
\]
Since for any $m\in M_{i}$ and $\omega\notin\Omega_{i}$, we have
$\sigma(m|\omega)=0$, the inequality above implies
\[
\sum_{\omega\in\Omega_{i}}\mu_{0}(\omega)\sum_{m\in M_{i}}\sigma(m|\omega)u_{R}(a_{i},\omega)\geq\sum_{\omega\in\Omega_{i}}\mu_{0}(\omega)\sum_{m\in M_{i}}\sigma(m|\omega)u_{R}(a',\omega)\quad\text{for all }a'\in A.
\]
Since $\sum_{m\in M_{i}}\sigma(m|\omega)=1$ if $\omega\in\Omega_{i}$,
we have
\begin{equation}
\sum_{\omega\in\Omega_{i}}\mu_{0}(\omega)u_{R}(a_{i},\omega)\geq\sum_{\omega\in\Omega_{i}}\mu_{0}(\omega)u_{R}(a',\omega)\quad\text{for all }a'\in A.\label{eq:R-BR-pure}
\end{equation}
To establish $(\hat{\sigma},\hat{\rho})$ is R-BR, we need to show
that for any $m_{i}\in M_{\hat{\sigma}}$, we have
\[
\sum_{\omega\in\Omega}\mu_{0}(\omega)\hat{\sigma}\left(m_{i}|\omega\right)\sum_{a\in A}\hat{\rho}\left(a|m_{i}\right)u_{R}(a,\omega)\geq\sum_{\omega\in\Omega}\mu_{0}(\omega)\hat{\sigma}\left(m_{i}|\omega\right)u_{R}(a',\omega)\quad\text{for all }a'\in A.
\]
But, by definition of $(\hat{\sigma},\hat{\rho})$, we know that $\hat{\sigma}\left(m_{i}|\omega\right)=0$
for $\omega\notin\Omega_{i}$ and that $\hat{\rho}\left(a_{i}|m_{i}\right)=1$.
Hence, the inequality above is equivalent to Equation $\eqref{eq:R-BR-pure}$.
\end{proof}

\subsection{\protect\label{subsec:Proof-of-Only-if-Theorem1}Proof of Theorem
\ref{Thm:commitmentWTA}}

Here we present and prove a result that generalizes Theorem \ref{Thm:commitmentWTA}
into a threefold equivalence. 
\begin{theoremprime}
Generically, the following statements are equivalent:
\begin{enumerate}
\item Commitment is valuable. 
\item Committed Sender values randomization. 
\item For any optimal persuasion profile $(\sigma,\rho)$, there exists
$m\in M_{\sigma}$ such that 
\[
|\arg\max_{a\in A}\sum_{\omega}\mu_{m}(\omega)u_{R}(a,\omega)|\geq2,
\]
where $\mu_{m}$ is defined as $\mu_{m}(\omega)=\frac{\mu_{0}(\omega)\sigma(m|\omega)}{\sum_{\omega}\mu_{0}(\omega)\sigma(m|\omega)}.$ 
\end{enumerate}
\end{theoremprime}
\begin{proof}
We establish the equivalence for any environment that satisfies both
partitional-unique-response and scant-indifferences. Since the set
of partitional-unique-response environments is generic (Lemma \ref{lemma:R-unique-response-genericity})
and the set of scant-indifferences environments is generic (Lemma
\ref{lemma:full-rank-genericity}), the set of environments that satisfy
both properties is also generic.

We will establish that (ii) implies (i), then that (i) implies (iii),
and finally that (iii) implies (ii).

Since we are in a scant-indifferences environment, (ii) implies (i)
by Proposition \ref{Lemma:purity}.

Next we wish to show that (i) implies (iii). We do so by establishing
the contrapositive. Suppose that there exists an optimal persuasion
profile $(\sigma,\rho)$ such that for every $m\in M_{\sigma}$, $\arg\max_{a\in A}\sum_{\omega}\mu_{m}(\omega)u_{R}(a,\omega)$
is unique. This implies that $\rho$ must be pure-on-path. We will
construct an optimal persuasion profile $(\sigma,\hat{\rho})$ that
is a cheap-talk equilibrium. Consider the following $\hat{\rho}$:
for all $m\in M_{\sigma}$, let $\hat{\rho}(m)=\rho(m)$; for $m\notin M_{\sigma}$,
let $\hat{\rho}(m)=\rho(m_{0})$ for some $m_{0}\in M_{\sigma}$.
Since $\hat{\rho}$ and $\rho$ coincide on path, $(\sigma,\rho)$
and $(\sigma,\hat{\rho})$ yield the same payoffs to both Sender and
Receiver. Therefore, $(\sigma,\hat{\rho})$ satisfies R-BR and yields
the persuasion payoff. It remains to show that $(\sigma,\hat{\rho})$
is S-BR, which is equivalent to Sender's interim optimality: for each
$\omega$, 
\begin{equation}
\sum_{m}\sigma(m|\omega)u_{S}(\hat{\rho}(m),\omega)\geq u_{S}(\hat{\rho}(m'),\omega)\label{eq:S-BR-interim}
\end{equation}
for all $m'\in M$. First, note that it suffices to show that Equation
$\left(\ref{eq:S-BR-interim}\right)$ holds for $m'\in M_{\sigma}$.
Once we establish that, we know $\sum_{m}\sigma(m|\omega)u_{S}(\hat{\rho}(m),\omega)\geq u_{S}(\hat{\rho}(m_{0}),\omega)$
since $m_{0}\in M_{\sigma}$. Therefore, since $\hat{\rho}(m')=\rho(m_{0})=\hat{\rho}(m_{0})$
for $m'\notin M_{\sigma}$, Equation $\left(\ref{eq:S-BR-interim}\right)$
holds for $m'\notin M_{\sigma}$. 

Now, suppose toward contradiction that there exist $\hat{\omega}$
and $\hat{m}\in M_{\sigma}$ such that $\sum_{m}\sigma(m|\hat{\omega})u_{S}(\hat{\rho}(m),\hat{\omega})<u_{S}(\hat{\rho}(\hat{m}),\hat{\omega})$.
Consider an alternative messaging strategy $\hat{\sigma}$: $\hat{\sigma}(\omega)=\sigma(\omega)$
for $\omega\neq\hat{\omega}$ while $\hat{\sigma}\left(\hat{\omega}\right)$
sends the same distribution of messages as $\sigma\left(\hat{\omega}\right)$
with probability $1-\varepsilon$ and otherwise sends message $\hat{m}$.
Formally, $\hat{\sigma}\left(m|\hat{\omega}\right)=\begin{cases}
\left(1-\varepsilon\right)\sigma\left(m|\hat{\omega}\right) & \text{if }m\neq\hat{m}\\
\left(1-\varepsilon\right)\sigma\left(m|\hat{\omega}\right)+\varepsilon & \text{if }m=\hat{m}
\end{cases}$.

Fix any $m\in M_{\sigma}$. Since $A$ is finite, the fact that $\hat{\rho}(m)=\rho\left(m\right)$
is the unique $\arg\max_{a\in A}\sum_{\omega}\mu_{m}(\omega)u_{R}(a,\omega)$
implies that $\hat{\rho}(m)$ remains the best response for a neighborhood
of beliefs around $\mu_{m}$. Therefore, for sufficiently small $\varepsilon$,
$(\hat{\sigma},\hat{\rho})$ is R-BR. Hence, $(\hat{\sigma},\hat{\rho})$
is a persuasion profile and yields the payoff 
\begin{align*}
U_{S}(\hat{\sigma},\hat{\rho}) & =U_{S}(\sigma,\hat{\rho})+\varepsilon\mu_{0}(\hat{\omega})[u_{S}(\hat{\rho}(\hat{m}),\hat{\omega})-\sum_{m}\sigma(m|\hat{\omega})u_{S}(\hat{\rho}(m),\hat{\omega})]\\
 & >U_{S}(\sigma,\hat{\rho}).
\end{align*}
This contradicts the fact that $(\sigma,\hat{\rho})$ yields the persuasion
payoff.

Finally, since we are considering a partitional-unique-response environment,
the fact that (iii) implies (ii) is immediate.
\end{proof}

\subsection{\protect\label{subsec:Proof-of-Theorem2}Proof of Theorem \ref{Thm:commitmentWTP}}

Lemma \ref{lemma:full-rank-genericity} and Proposition \ref{Lemma:purity}
jointly imply Theorem \ref{Thm:commitmentWTP}.

\subsection{\protect\label{subsec:Proof-of-Theorem3}Proof of Theorem \ref{Thm:how-often}}

Let $\lambda_{n}$ denote the Lebesgue measure on $\mathbb{R}^{n}$.
Recall that the set of environments is $[0,1]^{2|A||\Omega|}$. For
any property $p$ of an environment, let $\lambda_{2|A||\Omega|}(p):=\lambda_{2|A||\Omega|}(\{(u_{S},u_{R})|(u_{S},u_{R})\text{ satisfies }p\})$
denote the share of environments that satisfy $p$.

Given $u_{S}$, let $\Omega_{i}^{u_{S}}=\{\omega\in\Omega|a_{i}\in\arg\max_{a\in A}u_{S}(a,\omega)\}$
denote the set of states where $a_{i}$ is an ideal action for Sender.\footnote{In the body of the paper we denoted this set as $\Omega_{i}$, but
for the formal proofs, it is helpful to keep track of the fact that
this set depends on $u_{S}$.} Note that each $\omega$ must belong to at least one $\Omega_{i}^{u_{S}}$,
but the same $\omega$ may appear in multiple $\Omega_{i}^{u_{S}}$.
Say that $u_{S}$ is \emph{regular} if $\Omega_{i}^{u_{S}}\cap\Omega_{j}^{u_{S}}=\emptyset$
for $i\neq j$. We also say that an environment $\left(u_{S},u_{R}\right)$
is regular if $u_{S}$ is regular. Recall that an environment is felicitous
if Sender obtains his ideal payoff under cheap talk. We say that an
environment is regular-felicitous if it is regular and felicitous.
Lemmas \ref{lemma:full-rank-genericity} and \ref{lem:no-state-difference-between-two-actions}
jointly imply that $\lambda_{|A||\Omega|}(\{u_{S}\in[0,1]^{A\times\Omega}|u_{S}\text{ is regular\}})=1$,
so $\lambda_{2\left|A\right|\left|\Omega\right|}\left(\text{regular-felicity}\right)=\lambda_{2\left|A\right|\left|\Omega\right|}\left(\text{felicity}\right)$.

The next lemma provides a convenient characterization of regular-felicity.
\begin{lem}
\label{lem:Regular-felicity}A regular environment is felicitous if
and only if for each $i$,
\begin{equation}
\sum_{\omega\in\Omega_{i}^{u_{S}}}\mu_{0}\left(\omega\right)\left(u_{R}\left(a_{i},\omega\right)-u_{R}\left(a',\omega\right)\right)\geq0\quad\text{for all }a'\in A.\label{eq:felicitous-with-regularity}
\end{equation}
\end{lem}
\begin{proof}
Suppose that $u_{S}$ is regular and that \eqref{eq:felicitous-with-regularity}
holds. Select $|A|$ elements from $M$ and denote them by $m_{1}$
through $m_{|A|}$. Consider the following strategy profile $(\sigma,\rho)$:
\begin{itemize}
\item $\sigma(\omega)=m_{i}$ for every $\omega\in\Omega_{i}^{u_{S}}$;
\item $\rho(m)=a_{i}$ for $m=m_{i}$ ;
\item $\rho(m)=a_{1}$ for $m\notin\{m_{1},...,m_{|A|}\}$. 
\end{itemize}
From regularity and \eqref{eq:felicitous-with-regularity}, $\left(\sigma,\rho\right)$
is R-BR. In addition, in every state, Sender induces his ideal action,
so $\left(\sigma,\rho\right)$ is S-BR and yields Sender's ideal payoff.
Hence, the environment is felicitous.

Conversely, suppose that $u_{S}$ is regular and that Sender obtains
his ideal payoff in a cheap-talk equilibrium $(\sigma,\rho)$. By
regularity, Sender's ideal action in each state is unique. This implies
that $\rho$ must be pure on-path; otherwise, Sender would fail to
induce his unique ideal action with probability 1 in some state. 

For each $i$, let $M_{i}=\{m\in M_{\sigma}|\rho(m)=a_{i}\}$ be the
set of on-path messages that induce action $a_{i}$. For any $m\in M_{i}$,
R-BR implies that

\[
\sum_{\omega\in\Omega}\mu_{0}(\omega)\sigma(m|\omega)[u_{R}(\omega,a_{i})-u_{R}(\omega,a')]\geq0\quad\text{for all }a'\in A.
\]

Summing the inequality above over $m\in M_{i}$, we obtain

\[
\sum_{\omega\in\Omega}\sum_{m\in M_{i}}\mu_{0}(\omega)\sigma(m|\omega)[u_{R}(\omega,a_{i})-u_{R}(\omega,a')]\geq0\quad\text{for all }a'\in A.
\]

Since $(\sigma,\rho)$ yields Sender's ideal payoff, we have that
(i) for $\omega\in\Omega_{i}$,$\sum_{m\in M_{i}}\sigma(m|\omega)=1$;
and (ii) for $\omega\notin\Omega_{i}$, $\sigma(m|\omega)=0$ for
all $m\in M_{i}$. Therefore, the inequality simplifies to

\[
\sum_{\omega\in\Omega_{i}}\mu_{0}(\omega)[u_{R}(\omega,a_{i})-u_{R}(\omega,a')]\geq0\quad\text{for all }a'\in A.
\]
\end{proof}

\subsubsection{Arbitrary state space}

In this section, we establish that for any $\Omega$, the share of
environments such that commitment has no value is weakly greater than
$\frac{1}{\left|A\right|^{\left|A\right|}}.$ 
\begin{lem}
\label{lem:obedient-probability-bound} $\lambda_{2\left|A\right|\left|\Omega\right|}\left(\text{regular-felicity}\right)\geq\frac{1}{\left|A\right|^{\left|A\right|}}$.
\end{lem}
\begin{proof}
Fix some regular $u_{S}$. For any non-empty $\Omega_{i}^{u_{S}}$,
let 
\[
E_{i}=\{u_{R}\in[0,1]^{A\times\Omega_{i}^{u_{S}}}|\sum_{\omega\in\Omega_{i}^{u_{S}}}\mu_{0}(\omega)u_{R}(a_{i},\omega)>\sum_{\omega\in\Omega_{i}^{u_{S}}}\mu_{0}(\omega)u_{R}(a',\omega),\forall a'\neq a_{i}\}
\]
denote the set of Receiver's preferences on $A\times\Omega_{i}^{u_{S}}$
such that $a_{i}$ is Receiver's uniquely optimal action given the
information that $\omega\in\Omega_{i}^{u_{S}}$. By symmetry, we have
that $\lambda_{|A||\Omega_{i}^{u_{S}}|}(E_{i})=\frac{1}{|A|}$.\footnote{Note that the set of Receiver's preferences on $A\times\Omega_{i}^{u_{S}}$
such that Receiver is indifferent between two actions, given the information
that $\omega\in\Omega_{i}^{u_{S}}$, has measure zero.} The set of $u_{R}\in[0,1]^{A\times\Omega}$ such that $(u_{S},u_{R})$
is felicitous is $\prod_{i:\Omega_{i}^{u_{S}}\text{ is non-empty}}E_{i}$,
whose measure is
\begin{align}
\lambda_{\left|A\right|\left|\Omega\right|}\left(\prod_{i:\Omega_{i}^{u_{S}}\text{ is non-empty}}E_{i}\right) & =\prod_{i:\Omega_{i}^{u_{S}}\text{ is non-empty}}\lambda_{|A||\Omega_{i}^{u_{S}}|}(E_{i})\nonumber \\
 & =\prod_{i:\Omega_{i}^{u_{S}}\text{ is non-empty}}(1/|A|)\nonumber \\
 & \geq\prod_{i\in\left\{ 1,...,\left|A\right|\right\} }(1/|A|)\nonumber \\
 & =\frac{1}{|A|^{|A|}}.\label{eq:obedient-probability}
\end{align}
So, we have established that for any regular $u_{S}$, $\lambda_{\left|A\right|\left|\Omega\right|}(\{u_{R}|(u_{S},u_{R})\text{ is felicitous}\})\geq\frac{1}{|A|^{|A|}}.$
Recall that $\lambda_{\left|A\right|\left|\Omega\right|}(\{u_{S}|u_{S}\text{ is regular}\})=1$.
Therefore, 
\begin{align*}
\lambda_{2\left|A\right|\left|\Omega\right|}\left(\text{regular-felicity}\right) & =\int_{\{u_{S}|u_{S}\text{ is regular}\}}\int_{\{u_{R}\in[0,1]^{A\times\Omega}|(u_{S},u_{R})\text{ is felicitous}\}}d\lambda_{\left|A\right||\Omega|}\:d\lambda_{\left|A\right|\left|\Omega\right|}\\
 & \geq\int_{\{u_{S}|u_{S}\text{ is regular}\}}\frac{1}{|A|^{|A|}}d\lambda_{\left|A\right|\left|\Omega\right|}\\
 & =\frac{1}{|A|^{|A|}}.
\end{align*}
\end{proof}
Since commitment has no value in any felicitous environment, Lemma
\ref{lem:obedient-probability-bound} implies that 
\begin{align*}
\lambda_{2\left|A\right|\left|\Omega\right|}(\text{commitment has no value}) & \geq\frac{1}{\left|A\right|^{\left|A\right|}}.\\
\end{align*}

\subsubsection{Limit as $\left|\Omega\right|\rightarrow\infty$}

In this section, we establish that as $\left|\Omega\right|\rightarrow\infty$:
($*$) the share of felicitous environments converges to $\frac{1}{\left|A\right|^{\left|A\right|}}$,
and ($**$) generically, commitment has no value if and only if the
environment is felicitous. Together, these results imply part (ii)
of Theorem \ref{Thm:how-often}.

We first give an outline of the proof. We first establish part ($**$).
We show that generically, if the environment is jointly-inclusive,\footnote{Recall that an environment is jointly-inclusive if for every action
$a$, there is some state $\omega$ such that $a$ is the unique ideal
action for both Sender and Receiver in $\omega$.} then commitment having no value implies felicity (Lemma \ref{lem:permissive-implies-obedient}).
We then show that as $\left|\Omega\right|\rightarrow\infty$, the
share of jointly-inclusive preferences converges to one (Lemma \ref{lem:permissive-when-Omega-large}).
Combining these two results, we conclude that as $\left|\Omega\right|\rightarrow\infty$,
generically, commitment has no value if and only if the environment
is felicitous (Lemma \ref{lem:CNV->obedience}). To establish part
($*$), recall that $\lambda_{2\left|A\right|\left|\Omega\right|}\left(\text{felicity}\right)\geq\frac{1}{\left|A\right|^{\left|A\right|}}$
and that the reason this is an inequality is the possibility that
some $\Omega_{i}^{u_{S}}$ might be empty. We show that as $\left|\Omega\right|\rightarrow\infty$,
the share of preferences such that some $\Omega_{i}^{u_{S}}$ is empty
converges to zero, which implies that the share of felicitous preferences
converges to $\frac{1}{\left|A\right|^{\left|A\right|}}$ (Lemma \ref{lem:obedience-limit-probability}).
\begin{lem}
\label{lem:permissive-implies-obedient} If commitment has no value
in a jointly-inclusive environment that satisfies partitional-unique-response
and scant-indifferences, then the environment is felicitous. 
\end{lem}
\begin{proof}
Consider a jointly-inclusive environment that satisfies partitional-unique-response
and scant-indifferences and suppose that commitment has no value.
By Proposition \ref{Lemma:purity}, there is a partitional $\sigma$
and a pure strategy $\rho$ such that $\left|M_{\sigma}\right|\leq\left|A\right|$
and $(\sigma,\rho)$ is a cheap-talk equilibrium and yields the persuasion
payoff.

First, note that every action is induced under $(\sigma,\rho)$; that
is, for any $a\in A$, there exists $\omega$ such that $a=\rho(\sigma(\omega))$.
To see why, suppose toward contradiction that there is an $a^{*}\in A$
that is not induced. Since the environment is jointly-inclusive, there
exists $\omega^{*}$ such that 
\begin{equation}
u_{S}(a^{*},\omega^{*})>u_{S}(a,\omega^{*})\text{ and }u_{R}(a^{*},\omega^{*})>u_{R}(a,\omega^{*})\text{ for all }a\neq a^{*}.\label{eq:permissive}
\end{equation}
Since $\left|M_{\sigma}\right|\leq\left|A\right|<\left|M\right|$,
there is an unsent message, say $m^{*}$.

Consider the strategy profile $(\hat{\sigma},\hat{\rho})$: 
\begin{itemize}
\item $\hat{\sigma}(\omega)=\sigma(\omega)$ for $\omega\neq\omega^{*}$,
and $\hat{\sigma}\left(m|\omega^{*}\right)=\begin{cases}
\left(1-\varepsilon\right) & \text{if }m=\sigma(\omega^{*})\\
\varepsilon & \text{if }m=m^{*}\\
0 & \text{otherwise}
\end{cases}$.
\item $\hat{\rho}(m)=\rho(m)$ for $m\neq m^{*}$, and $\hat{\rho}(m^{*})=a^{*}$.
\end{itemize}
Note that $(\hat{\sigma},\hat{\rho})$ is R-BR for sufficiently small
$\varepsilon$. For any $m\notin\{\sigma(\omega^{*}),m^{*}\}$, Receiver's
belief upon observing $m$ is unchanged, so $\hat{\rho}(m)=\rho(m)$
remains a best response. For $m=m^{*}$, \eqref{eq:permissive} implies
that $\hat{\rho}(m^{*})=a^{*}$ is the best response. For $m=\sigma(\omega^{*})$,
the fact the environment satisfies partitional-unique-response implies
that $\hat{\rho}\left(m\right)=\rho(m)$ is the unique best response
to $\mu_{m}$. Moreover, since $A$ is finite, this further implies
that $\hat{\rho}(m)$ remains the best response for a neighborhood
of beliefs around $\mu_{m}$. Therefore, for sufficiently small $\varepsilon$,
$\hat{\rho}(m)$ remains a best response. 

Now, note that $\rho(\sigma(\omega^{*}))\neq a^{*}$ because $a^{*}$
is not induced under $(\sigma,\rho)$. By \eqref{eq:permissive},
\begin{align*}
U_{S}(\hat{\sigma},\hat{\rho}) & =U_{S}(\sigma,\rho)+\varepsilon\mu_{0}(\omega^{*})\left(u_{S}(a^{*},\omega^{*})-u_{S}(\rho(\sigma(\omega^{*})),\omega^{*})\right)\\
 & >U_{S}(\sigma,\rho).
\end{align*}
This contradicts the fact that $(\sigma,\rho)$ yields the persuasion
payoff. Hence, we have established that every action is induced under
$(\sigma,\rho)$. 

Next, we show that this fact, coupled with the maintained assumptions,
implies that the environment is felicitous. Recall that $(\sigma,\rho)$
is a cheap-talk equilibrium; hence for each $\omega$, 
\[
u_{S}(\rho(\sigma(\omega)),\omega)\geq u_{S}(\rho(m),\omega)\text{ for all }m\in M.
\]
Since every action is induced under $(\sigma,\rho)$, the inequality
above is equivalent to
\[
u_{S}(\rho(\sigma(\omega)),\omega)\geq u_{S}(a,\omega)\text{ for all }a\in A.
\]
Moreover, since the environment satisfies scant-indifferences, Lemma
\ref{lem:no-state-difference-between-two-actions} implies that the
environment is regular and
\begin{equation}
u_{S}(\rho(\sigma(\omega)),\omega)>u_{S}(a,\omega)\text{ for all }a\neq\rho(\sigma(\omega)).\label{eq:sender-ideal}
\end{equation}
Hence, $\Omega_{i}^{u_{S}}=\{\omega\in\Omega|\rho(\sigma(\omega))=a_{i}\}$
and $\Omega_{i}^{u_{S}}\cap\Omega_{j}^{u_{S}}=\emptyset$ for $i\neq j$.
Let $M_{i}=\{m\in M_{\sigma}|\rho(m)=a_{i}\}$. For each $i$ and
each $m\in M_{i}$, R-BR of $(\sigma,\rho)$ implies 
\[
\sum_{\omega\in\{\omega:\sigma(\omega)=m\}}\mu_{0}(\omega)u_{R}(a_{i},\omega)\geq\sum_{\omega\in\{\omega:\sigma(\omega)=m\}}\mu_{0}(\omega)u_{R}(a',\omega)\text{ for all }a'\in A.
\]
Summing over all $m\in M_{i}$, and noting that $\cup_{m\in M_{i}}\{\omega:\sigma(\omega)=m\}=\{\omega\in\Omega|\rho(\sigma(\omega))=a_{i}\}=\Omega_{i}^{u_{S}}$,
we have 
\[
\sum_{\omega\in\Omega_{i}^{u_{S}}}\mu_{0}(\omega)u_{R}(a_{i},\omega)\geq\sum_{\omega\in\Omega_{i}^{u_{S}}}\mu_{0}(\omega)u_{R}(a',\omega)\text{ for all }a'\in A.
\]
Hence, by Lemma \ref{lem:Regular-felicity}, the environment is felicitous.
\end{proof}
\begin{lem}
\label{lem:permissive-when-Omega-large} Fix $A$. As $|\Omega|\rightarrow\infty$,
$\lambda_{2\left|A\right|\left|\Omega\right|}\left(\text{joint-inclusivity}\right)\rightarrow1$. 
\end{lem}
\begin{proof}
For each $\omega\in\Omega$ and $a\in A$, let $E_{a,\omega}^{S}$
denote the set of Sender's preferences on $A\times\{\omega\}$ such
that action $a$ is Sender's uniquely ideal action in state $\omega$,
i.e., $E_{a,\omega}^{S}=\{u_{S}\in\left[0,1\right]^{A\times\left\{ \omega\right\} }|u_{S}(a,\omega)>u_{S}(a',\omega),\forall a'\neq a\}$.
Because the set of Sender's preferences on $A\times\{\omega\}$ such
that Sender is indifferent between two actions at $\omega$ has measure
zero, by symmetry we have that $\lambda_{\left|A\right|}\left(E_{a,\omega}^{S}\right)=1/|A|$
for each $\left(a,\omega\right)$. Similarly, $\lambda_{\left|A\right|}\left(E_{a,\omega}^{R}\right)=1/|A$
for each $\left(a,\omega\right)$ where $E_{a,\omega}^{R}=\{u_{R}\in[0,1]^{A\times\{\omega\}}|u_{R}(a,\omega)>u_{R}(a',\omega),\forall a'\neq a\}$.
Now, let $E_{a,\omega}=E_{a,\omega}^{S}\times E_{a,\omega}^{R}$ denote
the set of preferences on $A\times\{\omega\}$ such that $a$ is the
unique ideal action for both Sender and Receiver in state $\omega.$
By the product property of the Lebesgue measure, $\lambda_{2|A|}(E_{a,\omega})=1/|A|^{2}.$

Let $E_{a,\omega}^{c}\equiv[0,1]^{(A\times\{\omega\})^{2}}\backslash E_{a,\omega}$.
The Cartesian product $\prod_{\omega}E_{a,\omega}^{c}$ is the set
of environments in which action $a$ is not the unique ideal action
for both Sender and Receiver in any state. Let $E_{a}\equiv[0,1]^{(A\times\Omega)^{2}}\backslash\prod_{\omega}E_{a,\omega}^{c}$
denote the complement of $\prod_{\omega}E_{a,\omega}^{c}$, i.e.,
the set of environments in which action $a$ is the unique ideal action
for both Sender and Receiver in at least one state. Let $E_{a}^{c}\equiv[0,1]^{(A\times\Omega)^{2}}\backslash E_{a}$.
Note that
\begin{align*}
\lambda_{2\left|A\right|\left|\Omega\right|}(E_{a}) & =1-\lambda_{2\left|A\right|\left|\Omega\right|}\left(\prod_{\omega}E_{a,\omega}^{c}\right)\\
 & =1-\prod_{\omega}\lambda_{2|A|}\left(E_{a,\omega}^{c}\right)\\
 & =1-\prod_{\omega}\left(1-\frac{1}{|A|^{2}}\right)\\
 & =1-\left(1-\frac{1}{|A|^{2}}\right)^{|\Omega|}.
\end{align*}
The intersection $\cap_{a\in A}E_{a}$ is the set of jointly-inclusive
environments: for every action $a$, there exists at least one state
in which $a$ is the unique ideal action for both Sender and Receiver.
Therefore,
\begin{align*}
\lambda_{2\left|A\right|\left|\Omega\right|}(\text{\text{joint-inclusivity}})= & \lambda_{2\left|A\right|\left|\Omega\right|}(\cap_{a\in A}E_{a})\\
= & 1-\lambda_{2\left|A\right|\left|\Omega\right|}(\cup_{a\in A}E_{a}^{c})\\
\geq & 1-\sum_{a\in A}\lambda_{2\left|A\right|\left|\Omega\right|}(E_{a}^{c})\\
= & 1-\sum_{a\in A}\left(1-\lambda_{2\left|A\right|\left|\Omega\right|}(E_{a})\right)\\
= & 1-\sum_{a\in A}\left(1-\frac{1}{|A|^{2}}\right)^{|\Omega|}\\
= & 1-|A|\left(1-\frac{1}{|A|^{2}}\right)^{|\Omega|}\\
\rightarrow & 1\quad\text{as }|\Omega|\rightarrow\infty.
\end{align*}
\end{proof}
\begin{lem}
\label{lem:CNV->obedience} Fix $A$. As $|\Omega|\rightarrow\infty$,
generically, commitment has no value if and only if the environment
is felicitous.
\end{lem}
\begin{proof}
For any $\Omega$, commitment has no value if the environment is felicitous.
Hence, we only need to show that as $|\Omega|\rightarrow\infty$,
generically, commitment has no value only if the environment is felicitous.
Let \emph{CNV} denote the set of environments in which commitment
has no value, and \emph{F }denote the set of felicitous environments.
We need to show that as $|\Omega|\rightarrow\infty$, $\lambda_{2|A||\Omega|}(\text{CNV}\cap\text{F}^{c})\rightarrow0$.

Let \emph{JPS} denote the set of environments that are jointly-inclusive
and satisfy partitional-unique-response and scant-indifferences. We
know from Lemma \ref{lem:permissive-implies-obedient} that $\text{JPS}\cap\text{CNV}\cap\text{F}^{c}=\emptyset$,
so $\lambda_{2\left|A\right|\left|\Omega\right|}(\text{JPS}\cap\text{CNV}\cap\text{F}^{c})=0$.
As $|\Omega|\rightarrow\infty$, $\lambda_{2\left|A\right|\left|\Omega\right|}\left(\text{JPS}\right)\rightarrow1$
(Lemmas \ref{lemma:R-unique-response-genericity}, \ref{lemma:full-rank-genericity},
and \ref{lem:permissive-when-Omega-large}), which in turn implies
$\lambda_{2\left|A\right|\left|\Omega\right|}(\text{JPS}^{c}\cap\text{CNV}\cap\text{F}^{c})\rightarrow0$.
Therefore, as $|\Omega|\rightarrow\infty$, $\lambda_{2\left|A\right|\left|\Omega\right|}(\text{CNV}\cap\text{F}^{c})=\lambda_{2\left|A\right|\left|\Omega\right|}(\text{JPS}\cap\text{CNV}\cap\text{F}^{c})+\lambda_{2\left|A\right|\left|\Omega\right|}(\text{JPS}^{c}\cap\text{CNV}\cap\text{F}^{c})\rightarrow0.$
\end{proof}
\begin{lem}
\label{lem:obedience-limit-probability} Fix $A$. As $|\Omega|\rightarrow\infty$,
$\lambda_{2\left|A\right|\left|\Omega\right|}(\text{felicity})\rightarrow\frac{1}{|A|^{|A|}}$. 
\end{lem}
\begin{proof}
Since regular environments are generic, it will suffice to show that
$\lambda_{2\left|A\right|\left|\Omega\right|}(\text{regular-felicity})\rightarrow\frac{1}{|A|^{|A|}}$
as $|\Omega|\rightarrow\infty$. Let $E^{S}=\{u_{S}\in[0,1]^{A\times\Omega}|\Omega_{i}^{u_{S}}\text{ is non-empty for all }i\}$
and $E=\{(u_{S},u_{R})\in[0,1]^{(A\times\Omega)^{2}}|u_{S}\in E^{S}\}.$

As noted earlier in Equation $\left(\ref{eq:obedient-probability}\right)$,
for any regular $u_{S}\in E^{S}$, $\lambda_{\left|A\right|\left|\Omega\right|}(\{u_{R}|(u_{S},u_{R})\text{ is felicitous}\})=\prod_{i:\Omega_{i}^{u_{S}}\text{ is non-empty}}\frac{1}{|A|}=\frac{1}{|A|^{|A|}}$.
Therefore,
\begin{align*}
\lambda_{2\left|A\right|\left|\Omega\right|}(\text{regular-felicity}\cap E) & =\lambda_{2\left|A\right|\left|\Omega\right|}(\{(u_{S},u_{R})\in[0,1]^{(A\times\Omega)^{2}}|u_{S}\in E^{S},(u_{S},u_{R})\text{ is regular-felicitous})\\
 & =\int_{E^{S}}\int_{\{u_{R}|(u_{S},u_{R})\text{ is regular-felicitous}\}}d\lambda_{\left|A\right|\left|\Omega\right|}\:d\lambda_{\left|A\right|\left|\Omega\right|}\\
 & =\int_{E^{S}}\frac{1}{|A|^{|A|}}\,d\lambda_{\left|A\right|\left|\Omega\right|}\\
 & =\frac{1}{|A|^{|A|}}\lambda_{\left|A\right|\left|\Omega\right|}(E^{S})\\
 & =\frac{1}{|A|^{|A|}}\lambda_{2\left|A\right|\left|\Omega\right|}(E).
\end{align*}

Note that any jointly-inclusive environment must be contained in $E$.
Hence, by Lemma \ref{lem:permissive-when-Omega-large}, $\lambda_{2\left|A\right|\left|\Omega\right|}(E)\rightarrow1$
as $|\Omega|\rightarrow\infty.$ Therefore, $\lambda_{2\left|A\right|\left|\Omega\right|}(\text{regular-felicity})\rightarrow\frac{1}{|A|^{|A|}}$
as $|\Omega|\rightarrow\infty$.
\end{proof}
Lemmas \ref{lem:CNV->obedience} and \ref{lem:obedience-limit-probability}
jointly yield the fact that, as $|\Omega|\rightarrow\infty$, $\lambda_{2\left|A\right|\left|\Omega\right|}(\text{commitment has no value})\rightarrow\frac{1}{\left|A\right|^{\left|A\right|}}$.

\section{Online Appendix}

\subsection{\protect\label{subsec:finite}Role of finite $\Omega$ and $A$}

In this section, we illustrate some of the issues that arise when
$\Omega$ or $A$ is infinite.

\subsubsection{Infinite $\Omega$}

Suppose $A$ is finite but $\Omega$ is infinite with an atomless
prior. It is known (Corollary 1 in \citealt{zeng2023derandomization})
that in this case, randomization is never valuable. Yet, commitment
can be valuable. Hence, the analogue of Theorem \ref{Thm:commitmentWTA}
does not hold.

\subsubsection{Infinite $A$}

Consider a binary-state version of the quadratic-loss, constant-bias
preferences from \citet{crawford1982strategic}. Suppose $\Omega=\left\{ 0,1\right\} $
and let $\mu_{0}$ be equiprobable. We will contrast the case where
$A=\left[0,1\right]$ and the case where $A=\left\{ 0,\frac{1}{n},\frac{2}{n},...,\frac{n-1}{n},1\right\} $
for some even integer $n\geq2$. We refer to the former the \emph{interval
case }and the latter as the \emph{finite case}. 

Receiver's utility is $u_{R}\left(a,\omega\right)=-\left(a-\omega\right)^{2}$.
Sender's utility is $u_{S}\left(a,\omega\right)=-\left(a-\omega-b\right)^{2}$
for some $b>0$. We say a profile $(\sigma,\rho)$ is \emph{full revelation}
if there exist disjoint subsets $M_{0},M_{1}\subseteq M$ such that
$\sigma(M_{0}|\omega=0)=1$, $\sigma(M_{1}|\omega=1)=1$.

The parameter region where $b>\frac{1}{2}$ illustrates the contrast
between the interval case and the finite case.\footnote{When $b\leq\frac{1}{2}$, it does not matter whether we consider the
interval or the finite case. In both cases, full revelation yields
the persuasion payoff and is a cheap-talk equilibrium. Thus, neither
commitment nor randomization is valuable. } In the interval case, only a full revelation profile yields the persuasion
payoff, but such a profile cannot be a cheap-talk equilibrium.\footnote{Due to the strict convexity of Sender's indirect utility function,
the persuasion payoff can only be achieved by full revelation. However,
a full revelation profile cannot be a cheap-talk equilibrium: type
$\omega=0$ would deviate and send a message in $M_{1},$ because
$u_{S}(0,0)=-b^{2}<-(1-b)^{2}=u_{S}(1,0)$.} Hence, in the interval case, in contrast to Theorem \ref{Thm:commitmentWTA},
commitment has value even though committed Sender does not value randomization.
In the finite case, however, the assumptions underlying Theorem \ref{Thm:commitmentWTA}
apply. In this case, commitment is valuable and committed Sender values
randomization.\footnote{To see that committed Sender values randomization, note that a full
revelation profile yields the payoff $-[\mu_{0}(\omega=1)(0-0-b)^{2}+(1-\mu_{0}(\omega=1))(1-1-b)^{2}]=-b^{2}.$
Providing no information yields the payoff $-[\mu_{0}(\omega=1)(1/2-0-b)^{2}+(1-\mu_{0}(\omega=1))(1/2-1-b)^{2}]=-(b^{2}+1/4)$.
Therefore, Sender's partitional persuasion payoff is $-b^{2}$. But,
Sender can obtain a strictly higher payoff by inducing posteriors
$\mu=\frac{1}{2n}$ and $\mu=1$. When $\mu=\frac{1}{2n}$, Receiver's
optimal action is $a=\frac{1}{n}$, so Sender's interim value is $-[(1-\frac{1}{2n})(\frac{1}{n}-0-b)^{2}+(\frac{1}{2n})(\frac{1}{n}-1-b)^{2}]=-b^{2}+\frac{2b-1}{2n}>-b^{2}$.
When $\mu=1$, Sender's interim value is $-b^{2}$. Sender's ex-ante
payoff is a convex combination of the two interim values, which is
strictly higher than $-b^{2}.$}

\subsection{\protect\label{subsec:transparent}State-independent preferences}

As we mentioned in the discussion of related literature, several papers
examine value of commitment under the assumption that Sender has state-independent
preferences. To connect to that literature, it is worth asking whether
our results hold under that assumption. When $\left|A\right|\geq3$,
state-independent preferences by Sender mean that we are not in a
scant-indifferences environment, so the proofs above do not apply.
Nonetheless, Theorems \ref{Thm:commitmentWTA} and \ref{Thm:commitmentWTP}
indeed remain true.

To formalize this, say that environment $\left(u_{S},u_{R}\right)$
is transparent if there exists some function $v:A\rightarrow\mathbb{R}$
such that $u_{S}\left(a,\omega\right)=u_{S}(a,\omega')\equiv v(a)$
for any $a,\omega,\omega'$. The set of all transparent environments
is $[0,1]^{\left|A\right|\,\left(\left|\Omega\right|+1\right)}$.
A set of transparent environments is \emph{transparently-generic}
if it has Lebesgue measure one in $\mathbb{R}^{\left|A\right|\,\left(\left|\Omega\right|+1\right)}$.
We say that a claim holds \emph{generically in transparent environments},
if it holds for a transparently-generic set of environments.

A transparent environment $\left(u_{S},u_{R}\right)$ satisfies \emph{no-duplicate-actions}
if for any $a\neq a'$, $v(a)\neq v(a')$. Clearly, the set of no-duplicate-actions
transparent environments is transparently-generic.

A strategy profile $(\sigma,\rho)$ is a \emph{simple babbling cheap-talk
equilibrium }if it is a cheap-talk equilibrium in which $|M_{\sigma}|=1$
and $\rho(m)=a_{0}$ for all $m\in M$ and for some $a_{0}\in A$. 
\begin{lem}
\label{lem:transparent-commitment-no-value-implies-babbling}In a
no-duplicate-actions transparent environment, if commitment has no
value, then there exists a simple babbling cheap-talk equilibrium
that yields the persuasion payoff. 
\end{lem}
\begin{proof}
If commitment has no value, some cheap-talk equilibrium, denoted by
$\left(\sigma,\rho\right)$, yields the persuasion payoff. First,
we claim that $\rho$ must be pure on-path. Suppose by contradiction
that at some on-path message $m$, $|\Supp(\rho(\cdot|m))|>1$. R-BR
then implies that Receiver must be indifferent among all actions in
$\Supp(\rho(\cdot|m))$. Since the environment satisfies no-duplicate-actions,
Sender must strictly prefers one of the actions in $\Supp(\rho(\cdot|m))$
over all others. Therefore, an alternative strategy profile where
Receiver breaks ties in favor of Sender would still satisfy R-BR while
strictly improving Sender's payoff.

Since $\left(\sigma,\rho\right)$ is a cheap-talk equilibrium and
the environment is transparent, S-BR implies that for any $m,m'\in M_{\sigma}$,
we have $v(\rho(m))=v(\rho(m'))$. Moreover, because the environment
satisfies no-duplicate-actions, it follows that $\rho(m)=\rho(m')$
for any $m,m'\in M_{\sigma}$; that is, only a single action is induced
in equilibrium. 

We now construct a simple babbling cheap-talk equilibrium $(\hat{\sigma},\hat{\rho})$
that yields the same payoff as $\left(\sigma,\rho\right)$. Choose
an arbitrary message $m_{0}\in M_{\sigma}$. Let $\hat{\sigma}(m_{0}|\omega)=1$
for all $\omega\in\Omega$, and $\hat{\rho}(m)=\rho(m_{0})$ for all
$m\in M$. The strategy profile $(\hat{\sigma},\hat{\rho})$ trivially
satisfies S-BR and yields the same payoff as $\left(\sigma,\rho\right)$.
Since $\left(\sigma,\rho\right)$ satisfies R-BR, it follows that
for each $m\in M_{\sigma}$, $\sum_{\omega}\mu_{0}(\omega)\sigma(m|\omega)u_{R}(\rho(m),\omega)\geq\sum_{\omega}\mu_{0}(\omega)\sigma(m|\omega)u_{R}(a,\omega)$
for all $a\in A$. Summing over all $m$, we obtain
\[
\sum_{\omega\in\Omega}\mu_{0}(\omega)\sum_{m\in M}\sigma(m|\omega)u_{R}(\rho(m),\omega)\geq\sum_{\omega\in\Omega}\mu_{0}(\omega)\sum_{m\in M}\sigma(m|\omega)u_{R}(a,\omega)\quad\text{for all }a\in A.
\]
Since $\rho(m)=\rho(m_{0})$ for all $m\in M_{\sigma}$, we can rewrite
the inequality as
\[
\sum_{\omega\in\Omega}\mu_{0}(\omega)u_{R}(\rho(m_{0}),\omega)\geq\sum_{\omega\in\Omega}\mu_{0}(\omega)u_{R}(a,\omega)\quad\text{for all }a\in A,
\]
which then implies $(\hat{\sigma},\hat{\rho})$ satisfies R-BR. Therefore,
$(\hat{\sigma},\hat{\rho})$ is a simple babbling cheap-talk equilibrium
that yields the persuasion payoff.
\end{proof}
\begin{thm}
\label{Thm:transparent-committed-Sender}Generically in transparent
environments, commitment is valuable if and only if committed Sender
values randomization.
\end{thm}
\begin{proof}
We establish the equivalence for any transparent environment that
satisfies both partitional-unique-response and no-duplicate-actions.
Recall that whether an environment satisfies partitional-unique-response
does not depend on Sender\textquoteright s preferences, so the same
argument as in Lemma \ref{lemma:R-unique-response-genericity} implies
that the set of partitional-unique-response transparent environments
is transparently-generic. Moreover, the set of no-duplicate-actions
transparent environments is transparently-generic. Therefore, the
set of transparent environments that satisfy both properties is also
transparently-generic.

The same arguments that establish (i) implies (iii) and (iii) implies
(ii) in Theorem 1' apply directly to any partitional-unique-response
transparent environment, thereby proving the ``only if'' direction.
The ``if'' direction follows immediately from Lemma \ref{lem:transparent-commitment-no-value-implies-babbling}.
\end{proof}
\begin{thm}
\label{Thm:transparent-cheap-talk-Sender}Generically in transparent
environments, commitment is valuable if cheap-talk Sender values randomization.
\end{thm}
\begin{proof}
The theorem follows immediately from Lemma \ref{lem:transparent-commitment-no-value-implies-babbling}
and the fact that the set of no-duplicate-actions transparent environments
is transparently-generic.
\end{proof}

\subsection{\protect\label{subsec:Partial}Value of partial commitment}

We consider a partial commitment setting where Sender can commit to
a distribution of messages, as in \citet{lin2024credible}. For any
Sender's strategy $\sigma,$ let $D(\sigma):=\{\sigma':\Omega\rightarrow\Delta M|\sum_{\omega}\mu_{0}(\omega)\sigma'(m|\omega)=\sum_{\omega}\mu_{0}(\omega)\sigma(m|\omega),\forall m\}$
denote the set of messaging strategies that preserve the same distribution
of messages. 

We say a profile $(\sigma^{*},\rho^{*})$ is a \emph{curve equilibrium}
if 

\[
\sigma^{*}\in\argmax_{\sigma\in D(\sigma^{*})}U_{S}(\sigma,\rho^{*})
\]
\[
\rho^{*}\in\argmax_{\rho}U_{R}(\sigma^{*},\rho).
\]
The first condition requires that Sender has no incentive to deviate
to any other messaging strategy that preserves the same distribution
of messages as $\sigma^{*}$, and the second condition requires Receiver
to play a best response (i.e., R-BR).

The \emph{curve payoff} is the maximum $U_{S}$ induced by a curve
equilibrium. The \emph{curve partitional payoff} is the maximum $U_{S}$
induced by a partitional curve equilibrium. We say Sender \emph{values
committing to a curve} if the curve payoff is strictly higher than
the cheap-talk payoff, and that \emph{curve-committed Sender values
randomization} if the curve payoff is strictly higher than the curve
partitional payoff.

A function $u_{S}:A\times\Omega\rightarrow\mathbb{R}$ is \emph{strictly
supermodular} if there exists a total order $>_{A}$ on $A$ and a
total order $>_{\Omega}$ on $\Omega$ such that for $a'>_{A}a$ and
$\omega'>_{\Omega}\omega$,
\[
u_{S}(a',\omega')-u_{S}(a,\omega')>u_{S}(a',\omega)-u_{S}(a,\omega).
\]
To simplify notation, once we fix a strictly supermodular $u_{S}$,
we use $>$ in place of $>_{A}$ and $>_{\Omega}$.

A set of Receiver's preferences is generic if it has Lebesgue measure
one in $[0,1]^{\left|A\right|\,\left|\Omega\right|}$.

We say that an \emph{outcome distribution }$\pi:\Omega\rightarrow\Delta A$
is induced by a profile $(\sigma,\rho)$ if $\pi(a|\omega)=\sum_{m\in M}\rho(a|m)\sigma(m|\omega)$
for all $\omega,a$. The following lemma offers a characterization
of outcome distributions that can be induced by a curve equilibrium.
\begin{lem}
\label{lem:curve-characterization}Fix any strictly supermodular $u_{S}$.
An outcome distribution $\pi:\Omega\rightarrow\Delta A$ is induced
by some curve equilibrium $(\sigma,\rho)$ where $\rho$ is pure strategy
on-path if and only if 

1. $\pi$ is comonotone; that is, for any $a<a'$, if $\pi(a|\omega)>0\text{ and }\pi(a'|\omega')>0$,
we must have $\omega\leq\omega'$;

2. $\pi$ is $u_{R}$-obedient: for each $a,a'\in A$,
\[
\sum_{\omega\in\Omega}\pi(a|\omega)\mu_{0}(\omega)\left[u_{R}(a,\omega)-u_{R}(a',\omega)\right]\geq0.
\]
\end{lem}
\begin{proof}
The lemma follows immediately from Theorem 1 and Lemma 1 in \citet{lin2024credible},
with the small caveat that in \citet{lin2024credible}, $\rho$ is
restricted to be pure both on and off-path. However, note that off-path
actions do not affect whether a profile is a curve equilibrium. Thus,
the lemma follows.
\end{proof}
\begin{thm}
\label{Thm:Curve}Fix any strictly supermodular $u_{S}$. For a generic
set of Receiver's preferences, if Sender values committing to a curve,
then a curve-committed Sender values randomization.
\end{thm}
\begin{proof}
We consider any Receiver preference that satisfies partitional-unique-response.
By Lemma \ref{lemma:R-unique-response-genericity}, this set of preferences
is generic. 

We prove the statement by contraposition. Suppose that a curve-committed
Sender does not value randomization. This means there exists a partitional
curve equilibrium, denoted by $(\sigma,\rho)$, that yields the curve
payoff. Since $\sigma$ is partitional, partitional-unique-response
and R-BR of $(\sigma,\rho)$ imply that $\rho$ is a pure strategy
on-path. We will construct a strategy profile $(\sigma,\hat{\rho})$
that is a cheap-talk equilibrium and yields the curve payoff.

Consider the following $\hat{\rho}$: for all $m\in M_{\sigma}$,
let $\hat{\rho}(m)=\rho(m)$; for $m\notin M_{\sigma}$, let $\hat{\rho}(m)=\rho(m_{0})$
for some $m_{0}\in M_{\sigma}$. Since $\hat{\rho}$ and $\rho$ coincide
on path, $(\sigma,\rho)$ and $(\sigma,\hat{\rho})$ induce the same
outcome distribution and yield the same payoffs to both Sender and
Receiver. Therefore, $(\sigma,\hat{\rho})$ satisfies R-BR and yields
the curve payoff. It remains to show that $(\sigma,\hat{\rho})$ is
S-BR, which is equivalent to Sender's interim optimality: for each
$\omega$, 
\begin{equation}
u_{S}(\hat{\rho}(\sigma(\omega)),\omega)\geq u_{S}(\hat{\rho}(m'),\omega)\label{eq:S-BR-interim-curve}
\end{equation}
for all $m'\in M$. Note that it suffices to show that Equation $\left(\ref{eq:S-BR-interim-curve}\right)$
holds for $m'\in M_{\sigma}$. Once we establish that, we know $\sum_{m}\sigma(m|\omega)u_{S}(\hat{\rho}(m),\omega)\geq u_{S}(\hat{\rho}(m_{0}),\omega)$
since $m_{0}\in M_{\sigma}$. Therefore, since $\hat{\rho}(m')=\rho(m_{0})=\hat{\rho}(m_{0})$
for $m'\notin M_{\sigma}$, Equation $\left(\ref{eq:S-BR-interim-curve}\right)$
holds for $m'\notin M_{\sigma}$. 

Since $(\sigma,\rho)$ is a curve equilibrium and $\rho$ is pure
on-path, by Lemma \ref{lem:curve-characterization}, the induced outcome
distribution, denoted by $\pi$, satisfies comonotonicity and $u_{R}$-obedience.
In addition, since $\sigma$ is partitional and $\rho$ is pure on-path,
the induced outcome distribution $\pi$ is also a pure mapping. Moreover,
$\pi$ being comonotone can be strengthened to $\pi$ being monotone
partitional:
\begin{equation}
\forall a<a',\text{ }\pi(a|\omega)>0\text{ and }\pi(a'|\omega')>0\text{ implies }\omega<\omega'.\label{eq:curve-monotone-partition}
\end{equation}

Moreover, since the environment satisfies partitional-unique-response,
$u_{R}$-obedience can be strengthened to strict $u_{R}$-obedience:
for each $a\in A^{*}\equiv\cup_{\omega\in\Omega}\Supp(\pi(\cdot|\omega))$
and $a'\in A/\{a\}$,
\begin{equation}
\sum_{\omega\in\Omega}\pi(a|\omega)\mu_{0}(\omega)\left[u_{R}(a,\omega)-u_{R}(a',\omega)\right]>0.\label{eq:curve-strict-uR-obedience}
\end{equation}
For each $a\in A^{*}$, let $\Omega_{a}=\{\omega|\hat{\rho}(\sigma(\omega))=a\}$
denote the set of states that induce action $a$. By \eqref{eq:curve-monotone-partition},
$\{\Omega_{a}\}_{a\in A^{*}}$ forms a monotone partition of $\Omega$:
$\cup_{a\in A^{*}}\Omega_{a}=\Omega$, and for any $a<a'$, $\omega\in\Omega_{a}$,
$\omega'\in\Omega_{a'}$, we have $\omega<\omega'$.

To establish Equation $\left(\ref{eq:S-BR-interim-curve}\right)$,
it suffices to show that for each $\omega\in\Omega$, $u_{S}(a',\omega)\leq u_{S}(\pi(\omega),\omega)$
for all $a'\in A^{*}.$ Suppose, toward a contradiction, that there
exists $\omega^{*}\in\Omega$ and $a'\in A^{*}$ such that $u_{S}(a',\omega^{*})>u_{S}(\pi(\omega^{*}),\omega^{*})$.
Without loss of generality, we assume $a'>\pi(\omega^{*})$; the proof
for $a'<\pi(\omega^{*})$ follows symmetrically.

Let $a^{*}\equiv\pi(\omega^{*})$ and $\hat{a}\in\min\{\argmax_{a'>a^{*},a'\in A^{*}}u_{S}(a',\omega)\}$
denote type $\omega^{*}$'s smallest optimal action among $\{a'|a'>a^{*}\}$.
Let $\tilde{\omega}=\max\{\omega|\pi(\omega)<\hat{a}\}$ denote the
largest type that induces an action smaller than $\hat{a}$. Let $\tilde{a}\equiv\pi(\tilde{\omega})<\hat{a}$.
Since $\hat{a}$ is $\omega^{*}$'s smallest optimal action among
$\{a'|a'>a^{*}\}$, $u_{S}(\hat{a},\omega^{*})>u_{S}(\tilde{a},\omega^{*}).$
By definition, $\tilde{\omega}\geq\omega^{*}$, so by supermodularity,
\begin{equation}
u_{S}(\hat{a},\tilde{\omega})-u_{S}(\tilde{a},\tilde{\omega})\geq u_{S}(\hat{a},\omega^{*})-u_{S}(\tilde{a},\omega^{*})>0.\label{eq:curve-strict-value-improvement}
\end{equation}

We now construct an alternative outcome distribution $\hat{\pi}$:
$\hat{\pi}(\omega)=\pi(\omega)$ for $\omega\neq\tilde{\omega}$,
while $\hat{\pi}\left(\tilde{\omega}\right)$ induces action $\tilde{a}$
with probability $1-\varepsilon$ and induces action $\hat{a}$ with
probability $\varepsilon$. By \eqref{eq:curve-strict-value-improvement},
$\hat{\pi}$ yields a strictly higher value than $\pi$. In addition,
by \eqref{eq:curve-strict-uR-obedience}, for sufficiently small $\varepsilon$,
$\hat{\pi}$ remains obedient. 

Lastly, we show that $\hat{\pi}$ satisfies comonotonicity. To see
this, first note that whether an outcome distribution $\hat{\pi}$
satisfies comonotonicity depends only on its support; that is, the
set of $(a,\omega)$ such that $\hat{\pi}(a|\omega)>0$. By construction,
the supports of $\hat{\pi}$ and $\pi$ differ only in that $\hat{\pi}$'s
support contains an additional element, $(\tilde{\omega},\hat{a}).$
Since $\pi$ is comonotone, to establish that $\hat{\pi}$ is comonotone,
it suffices to show that: for any $a<\hat{a}$ and $\omega\in\Omega_{a}$,
we have $\omega\leq\tilde{\omega}$; for any $a'>\hat{a}$ and $\omega'\in\Omega_{a'}$,
we have $\omega'\geq\tilde{\omega}$. To prove the first part, note
that $\{\Omega_{a}\}_{a\in A^{*}}$ forms a monotone partition of
$\Omega$, which implies that for any $a<\hat{a}$ and $\omega\in\Omega_{a}$,
we have $\pi(\omega)<\hat{a}$. Recall that $\tilde{\omega}=\max\{\omega|\pi(\omega)<\hat{a}\}$
is largest type that induces an action smaller than $\hat{a}$; therefore,
$\omega\leq\tilde{\omega}.$ To prove the second part, note that since
$\{\Omega_{a}\}_{a\in A^{*}}$ forming a monotone partition, it follows
that that for any $a'>\hat{a}$, and $\omega'\in\Omega_{a'}$, we
have $\omega'>\min\Omega_{\hat{a}}>\tilde{\omega}$.

Since $\hat{\pi}$ that satisfies comonotonicity and $u_{R}$-obedience,
by Lemma \ref{lem:curve-characterization}, there exists a curve equilibrium
that yields a strictly higher payoff than $(\sigma,\rho)$. This contradicts
the fact that $(\sigma,\rho)$ yields the curve payoff.
\end{proof}

\subsection{Away from zeros}

Theorem \ref{Thm:commitmentWTA} tells us that, generically, commitment
has zero value if and only if randomization has zero value. A natural
question is whether, generically, a small value of commitment implies
or is implied by a small value of randomization. This section establishes
that the answer is no.

Specifically, we show that there is a $\Delta>0$ such that for any
$\varepsilon>0$: (i) there is a positive measure of environments
where the value of commitment is greater than $\Delta$ and the value
of randomization is less than $\varepsilon$, and (ii) there is a
positive measure of environments where the value of randomization
is greater than $\Delta$ and the value of commitment is less than
$\varepsilon$.\footnote{Recall that we assume that the range of the utility functions is $\left[0,1\right]$
in order to succinctly state our results about how often commitment
is valuable. If we relax the assumption that the range of the utility
functions is uniformly bounded, we can alternatively show that \emph{for
any} $\Delta>0$ and any $\varepsilon>0$: (i) there is a positive
measure of environments where the value of commitment is greater than
$\Delta$ and the value of randomization is less than $\varepsilon$,
and (ii) there is a positive measure of environments where the value
of randomization is greater than $\Delta$ and the value of commitment
is less than $\varepsilon$.} 

We begin by illustrating the role of the genericity condition for
Theorem \ref{Thm:commitmentWTA}. We present two examples. The first
example presents an environment (that violates partitional-unique-response)
where commitment is valuable but randomization is not. The second
example presents an environment (that violates scant indifferences)
where randomization is valuable but commitment is not.

Then, we build on the first example to show to construct, given any
$\varepsilon$, a positive measure of environments where the value
of commitment is above $\frac{1}{6}$ but the value of randomization
is less than $\varepsilon$. We build on the second example to show
how to construct, given any $\varepsilon$, a positive measure of
environments where the value of randomization is above $\frac{1}{13}$
but the value of of commitment is less than $\varepsilon$.
\begin{example}
\label{Ex:value-commitment-not-value-randomization}Consider $\Omega=\{\omega_{1},\omega_{2}\}$
with prior $\mu_{0}(\omega_{1})=\mu_{0}(\omega_{2})=0.5$ and $A=\{a_{1},a_{2}\}$.
Players' payoffs are given in Table \ref{table: example-partition-unique}.
Receiver's best response is $a_{1}$ when $\mu\equiv\mu(\omega_{2})\in[0,1)$,
and she is indifferent between $a_{1}$ and $a_{2}$ when $\mu=1$.
The concavification of Sender's indirect utility function is depicted
in Figure \ref{figure: example-partition-unique}.

\begin{figure}[h]  \centering
\begin{minipage}[b]{0.45\textwidth} \centering

\begin{center}
\begin{tabular}{|c|c|c|c|}
\hline
    $u_S$     &  $a_1$  &  $a_2$     \\ \hline
$\omega_1$ &  0   &  $2/3$      \\ \hline
$\omega_2$ &  0    &  $2/3$     \\ \hline
\end{tabular}
\end{center}

\vspace{0.05in}
\begin{center}
\begin{tabular}{|c|c|c|}
\hline
    $u_R$      &  $a_1$  &  $a_2 $    \\ \hline
$\omega_1$ &  1/2   &  0         \\ \hline
$\omega_2$ &  1/2    &  1/2        \\ \hline
\end{tabular}
\end{center}
\vspace{-0.2in}
\end{minipage} 
\hspace{3ex}
\begin{minipage}[b]{0.4\textwidth} 
\centering
\begin{tikzpicture}[domain=0:3, scale=5, thick]

\draw[<->] (1.1,0)node[right]{$\mu$}--(0,0)node[left,yshift=-5]{0}--(0,0.9)node[above]{${v}(\mu)$};

\draw (1.02,0)node[below]{1};

\draw[blue, ultra thick] (0,0.01)--(1,0.01);

\filldraw[blue] (1,2/3) circle (0.3pt);

\draw[red,thick] (0,0+0.01)--(1,2/3);

\draw[ dashed] (0.5,0)node[below]{$\mu_0=0.5$}--(0.5,1/3);



\draw (0,2/3)node[left]{$2/3$}--(0.02,2/3);

\end{tikzpicture}
\vspace{-0.3in}
\end{minipage}
\vspace{0.4in}

\begin{minipage}[t]{0.45\textwidth}
\captionof{table}{Sender and Receiver's payoffs \label{table: example-partition-unique}}    
\end{minipage}
\begin{minipage}[t]{0.4\textwidth}
\caption{Concavification\label{figure: example-partition-unique}}
\end{minipage}
\end{figure}

Clearly, full revelation is the unique optimal information structure,
which yields a payoff of $1/3$; therefore, Sender does not value
randomization. In addition, the only possible cheap-talk equilibrium
outcome is babbling, which yields a payoff of 0. Hence, commitment
is valuable.
\end{example}
\begin{example}
\label{Ex:value-randomization-not-value-commitment}Consider $\Omega=\{\omega_{1},\omega_{2}\}$
with prior with prior $\mu_{0}(\omega_{1})=\mu_{0}(\omega_{2})=0.5$
and $A=\{a_{1},a_{2},a_{3}\}$. Players' payoffs are given in Table
\ref{table: example-scant-indifference}. Receiver's best response
is $a_{1}$ when $\mu\in[0,1/3]$, $a_{2}$ when $\mu\in[1/3,2/3]$,
and $a_{3}$ when $\mu\equiv\mu(\omega_{2})\in[2/3,1]$. This leads
to Sender's indirect utility function (blue) and its concave envelope
(red) depicted in Figure \ref{figure: example-scant-indifference}.

\begin{figure}[h]  \centering
\begin{minipage}[b]{0.45\textwidth} \centering

\begin{center}
\begin{tabular}{|c|c|c|c|}
\hline
    $u_S$       &  $a_1$  &  $a_2$  &  $a_3$    \\ \hline
$\omega_1$ &  $1/3$   &  $2/3$    &  $0$      \\ \hline
$\omega_2$ &  $5/6$    &  $1/6$    &  $1/2$    \\ \hline
\end{tabular}
\end{center}

\vspace{0.05in}
\begin{center}
\begin{tabular}{|c|c|c|c|}
\hline
    $u_R$       &  $a_1$  &  $a_2$  & $ a_3$    \\ \hline
$\omega_1$ &  $3/4$   &  $1/2$    &  $0$      \\ \hline
$\omega_2$ &  $0$    &  $1/2$    &  $3/4$    \\ \hline
\end{tabular}
\end{center}
\vspace{-0.2in}
\end{minipage} 
\hspace{3ex}
\begin{minipage}[b]{0.4\textwidth} 
\centering
\begin{tikzpicture}[domain=0:3, scale=5, thick]

\draw[<->] (1.1,0)node[right]{$\mu$}--(0,0)node[left,yshift=-5]{0}--(0,0.9)node[above]{${v}(\mu)$};

\draw (1.02,0)node[below]{1};

\draw[blue] (0,1/3)--(1/3,1/2)  (1/3,1/2)--(2/3, 1/3) (2/3, 1/3)--(1, 1/2);

\draw[red,thick] (0,1/3+0.01)--(1/3,1/2+0.01)--(1, 1/2+0.01);

\draw[ dashed] (0.5,0)node[below]{$\mu_0=0.5$}--(0.5,1/2);

\draw (1/6,0)node[above]{\small$a_1$} (1/2,0)node[above]{\small$a_2$} (5/6,0)node[above,xshift=2]{\small$a_3$} ;

\draw (1/3,0)--(1/3,0.02) (2/3,0)--(2/3,0.02)  (1,0)--(1,0.02);

\draw (0,1/2)node[left]{$1/2$}--(0.02,1/2);

\end{tikzpicture}
\vspace{-0.3in}
\end{minipage}
\vspace{0.4in}

\begin{minipage}[t]{0.45\textwidth}
\captionof{table}{Sender's and Receiver's payoffs \label{table: example-scant-indifference}}    
\end{minipage}
\begin{minipage}[t]{0.4\textwidth}
\caption{Concavification\label{figure: example-scant-indifference}}
\end{minipage}
\end{figure}

Sender values randomization, because the unique optimal information
structure that induces a posterior $1/3$ cannot be generated by a
partitional messaging strategy. Since Sender's indirect utility function
is continuous, by \citet{lipnowski2020equivalence}, Sender does not
value commitment. For the sake of completeness, we construct a cheap-talk
equilibrium that yields the persuasion payoff $1/2$ to Sender.

Consider a strategy profile $(\sigma,\rho)$ with two on-path messages
$m_{1},m_{2}$: $\sigma(m_{1}|\omega_{2})=1/2$, $\sigma(m_{2}|\omega_{2})=1/2$,
$\sigma(m_{1}|\omega_{1})=1$; $\rho(a_{1}|m_{1})=1/2$, $\rho(a_{2}|m_{1})=1/2$
$\rho(a_{3}|m_{2})=1.$

The profile satisfies R-BR because the posterior upon observing $m_{1}$
is 1/3 and when observing $m_{2}$ is $1$. We next show that the
profile also satisfies S-BR. For type $\omega_{2}$ Sender, the expected
payoff of sending message $m_{2}$ is 1/2 and the expected payoff
of sending message $m_{1}$ is $\frac{1}{2}\cdot\frac{5}{6}+\frac{1}{2}\cdot\frac{1}{6}=\frac{1}{2}$,
so type $\omega_{2}$ Sender is indifferent and has no incentive to
deviate. For type $\omega_{1}$ Sender, the expected payoff of sending
message $m_{2}$ is $0$ and the expected payoff of sending message
$m_{1}$ is $\frac{1}{2}\cdot\frac{1}{3}+\frac{1}{2}\cdot\frac{2}{3}=\frac{1}{2}$,
so he strictly prefer sending message $m_{1}$. Hence, $(\sigma,\rho)$
is a cheap-talk equilibrium that yields the persuasion payoff.
\end{example}

\subsubsection{\protect\label{subsec:Large-comm-small-rand}Small value of randomization
does not imply small value of commitment}

We construct a positive measure of environments where the value of
commitment is bounded away from zero but the value of randomization
is arbitrarily small. Formally, given an environment $(u_{S},u_{R})$,
let $\Delta_{\mathcal{R}}(u_{S},u_{R})=\text{\emph{Persuasion Payoff}}-\text{\emph{Partitional Persuasion Payoff}}$
and $\Delta_{\mathcal{C}}(u_{S},u_{R})=\text{\emph{Persuasion Payoff}}-\text{\emph{Cheap-Talk Payoff}}$
denote the value of randomization and the value of commitment, respectively.
We will establish that for any $\varepsilon>0$, there exists a positive
measure set of environments $E$ such that for any $(u_{S},u_{R})\in E$,
we have $\Delta_{\mathcal{R}}(u_{S},u_{R})<\varepsilon$ and $\Delta_{\mathcal{C}}(u_{S},u_{R})>1/6$.

Fix any $\varepsilon>0$, we perturb players' payoffs in Example \ref{Ex:value-commitment-not-value-randomization}
to construct the desired positive share of environments. The idea
is that for sufficiently small perturbations in an appropriate direction,
the changes to the persuasion payoff, partitional persuasion payoff,
and cheap-talk payoff will also be small. Since in Example \ref{Ex:value-commitment-not-value-randomization},
the value of randomization is zero and the value of commitment is
positive, we obtain a positive measure of environments with a small
value of randomization and a value of commitment bounded away from
zero.

Players' payoffs are as in Table \ref{table:small-random-large-commit},
where $s_{ij},r_{ij}$ are the perturbations to Sender's and Receiver's
payoffs, respectively, when action $a_{j}$ is taken in state $\omega_{i}$.

\begin{table}
\begin{center}
\begin{tabular}{|C|S|S|}
\hline
    u_S    &  a_1  &  a_2    \\ \hline
\omega_1&  0+s_{11}   &  2/3+s_{12}      \\ \hline
\omega_2 &  0+s_{21}    &  2/3+s_{22}    \\ \hline
\end{tabular}
\end{center}

\vspace{0.05in}
\begin{center}
\begin{tabular}{|C|S|S|}
\hline
    u_R      &  a_1  &  a_2     \\ \hline
\omega_1 &  1/2+r_{11}  &  r_{12}      \\ \hline
\omega_2 &  1/2+r_{21}   &  1/2+r_{22}       \\ \hline
\end{tabular}
\end{center}
\caption{Sender's and Receiver's payoffs}\label{table:small-random-large-commit}
\end{table}

Let $s_{ij}\in[0,\delta]$, $r_{11},r_{21}\in[-\delta,0]$, and $r_{12},r_{22}\in[0,\delta]$,
where $\delta>0$. These perturbations generate a positive measure
set of environments, denoted by $E^{\delta}$. We will establish that
if $\delta<\min\{\frac{1}{4},\frac{\varepsilon}{2(1+\varepsilon)}\},$
then for any $(u_{S},u_{R})\in E^{\delta},$ the value of commitment
is greater than $1/6$ and the value of randomization is less than
$\varepsilon$. 

Consider any $(u_{S},u_{R})\in E^{\delta}$. Since $\delta<1/4,$
Receiver's best response is $a_{2}$ iff $\mu\geq\mu^{*}\equiv\frac{1/2+r_{11}-r_{12}}{1/2+r_{11}-r_{21}-r_{12}+r_{22}}\in[1-2\delta,1].$
Sender's indirect utility function is therefore
\[
v(\mu)=\begin{cases}
s_{11}+(s_{21}-s_{11})\mu & \text{if }\mu<\mu^{*}\\
\frac{2}{3}+s_{12}+(s_{22}-s_{12})\mu & \text{if }\mu\geq\mu^{*}.
\end{cases}
\]

By inducing beliefs $\mu=0$ and $\mu=\mu^{*}$, Sender achieves her
persuasion payoff
\[
\frac{1}{2\mu^{*}}[\frac{2}{3}+s_{12}+(s_{22}-s_{12})\mu^{*}]+(1-\frac{1}{2\mu^{*}})s_{11}.
\]
Meanwhile, full revelation yields a payoff of 
\[
\frac{1}{3}+\frac{s_{11}+s_{22}}{2}.
\]
Therefore, 
\begin{align*}
\Delta_{\mathcal{R}}(u_{S,}u_{R})\leq & \frac{1}{2\mu^{*}}[\frac{2}{3}+s_{12}+(s_{22}-s_{12})\mu^{*}]+(1-\frac{1}{2\mu^{*}})s_{11}-\frac{1}{3}-\frac{s_{11}+s_{22}}{2}\\
= & \frac{1-\mu^{*}}{\mu^{*}}(\frac{1}{3}+\frac{s_{12}-s_{11}}{2})\\
\leq & \frac{2\delta}{1-2\delta}(\frac{1}{3}+\delta)\\
< & \varepsilon
\end{align*}
where the third line follows from $\mu^{*}\geq1-2\delta$ and $s_{ij}\in[0,\delta]$,
and the last line follows from $\delta<\frac{\varepsilon}{2(1+\varepsilon)}$
and $\delta<\frac{2}{3}$.

In addition, since $\delta<1/4$, $a_{2}$ is Sender's preferred action
regardless of the states. It follows that any cheap-talk equilibrium
outcome must be the babbling outcome where Receiver takes action $a_{1}$
with probability 1. Therefore, Sender's cheap-talk payoff is $\frac{s_{11}+s_{21}}{2}$.
Hence, $\Delta_{\mathcal{C}}(u_{S,}u_{R})=\text{\emph{Persuasion Payoff}}-\text{\emph{Cheap-Talk Payoff}}\geq\frac{1}{3}+\frac{s_{11}+s_{22}}{2}-\frac{s_{11}+s_{21}}{2}>\frac{1}{6}$,
where the weak inequality follows from the fact that the persuasion
payoff is greater than the payoff from full revelation, and the strict
inequality follows from $\frac{s_{22}-s_{21}}{2}\geq\frac{-\delta}{2}>-\frac{1}{6}$.

\subsubsection{\protect\label{subsec:Small-comm-large-rand}Small value of commitment
does not imply small value of randomization}

We will establish that for any $\varepsilon>0$, there exists a positive
measure set of environments $E$ such that for any $(u_{S},u_{R})\in E$,
we have $\Delta_{\mathcal{C}}(u_{S},u_{R})<\varepsilon$ and $\Delta_{\mathcal{R}}(u_{S},u_{R})>\frac{1}{13}$. 

Fix any $\varepsilon>0$, we perturb players' payoffs in Example \ref{Ex:value-randomization-not-value-commitment}
to construct the desired positive share of environments. Similar to
Section \ref{subsec:Large-comm-small-rand}, the idea is to show that
changes to the persuasion payoff, partitional persuasion payoff, and
cheap-talk payoff are small under small perturbations. Since in Example
\ref{Ex:value-randomization-not-value-commitment}, the value of commitment
is zero and the value of randomization is positive, we obtain a positive
measure of environments with a small value of commitment and a value
of randomization bounded away from zero.

Players' payoffs are as in Table \ref{table:small-commit-large-random},
where $s_{ij},r_{ij}$ are the perturbations to Sender's and Receiver's
payoffs, respectively, when action $a_{j}$ is taken in state $\omega_{i}$.

\begin{table}
\begin{center}
\begin{tabular}{|C|N|N|N|}
\hline
    u_S      &  a_1  &  a_2  &  a_3    \\ \hline
\omega_1 &  1/3+s_{11}   &  2/3 +s_{12}   &  0+s_{13}      \\ \hline
\omega_2 &  5/6+s_{21}  &  1/6 +s_{22}   &  1/2+s_{23}    \\ \hline
\end{tabular}
\end{center}

\vspace{0.05in}
\begin{center}
\begin{tabular}{|C|N|N|N|}
\hline
    u_R       &  a_1  &  a_2  &  a_3    \\ \hline
\omega_1 &  3/4+r_{11}   &  1/2+r_{12}    & 0+r_{13}     \\ \hline
\omega_2 &  0+r_{21}    &  1/2+r_{22}    &  3/4+r_{23}    \\ \hline
\end{tabular}
\end{center}
\caption{Sender's and Receiver's payoffs}\label{table:small-commit-large-random}
\end{table}

Let $s_{ij},r_{ij}\in[0,\delta]$, where $\delta>0$. These perturbations
generate a positive measure set of environments, denoted by $E^{\delta}$.
Consider any $(u_{S},u_{R})\in E^{\delta}$ with $\delta<1/4$. Receiver's
best response is 
\[
a_{R}(\mu)=\begin{cases}
a_{1} & \text{if }\mu\in[0,\mu_{12}]\\
a_{2} & \text{if }\mu\in[\mu_{12},\mu_{23}]\\
a_{3} & \text{if }\mu\in[\mu_{23},1]
\end{cases}
\]
where $\mu_{12}=\frac{1}{3+r_{12}+r_{22}-r_{12}-r_{21}}$ and $\mu_{23}=\frac{2+r_{12}-r_{13}}{3+r_{23}-r_{13}+r_{12}-r_{22}}.$
Similar to Example \ref{Ex:value-randomization-not-value-commitment}
, the optimal information structure induces beliefs $\mu_{12}$ and
$1$, yielding a value
\begin{multline*}
\frac{1}{2(1-\mu_{12})}\max\{\mu_{12}(5/6+s_{21})+(1-\mu_{12})(1/3+s_{11}),\mu_{12}(1/6+s_{22})+(1-\mu_{12})(2/3+s_{12})\}\\
+\frac{1-2\mu_{12}}{2(1-\mu_{12})}(1/2+s_{23}).
\end{multline*}

Since $s_{ij}\in[0,\delta],$ taking $\delta\rightarrow0$, we have
$\mu_{12}\rightarrow1/3$, and the above value approaches $1/2$.
By continuity, there exists $\delta^{1}>0$ such that for any $\delta<\delta^{1}$,
the persuasion value lies within the interval $[\frac{1}{2}-\frac{\varepsilon}{2},\frac{1}{2}+\frac{\varepsilon}{2}].$ 

Meanwhile, full revelation yields a payoff of $\frac{5/6+s_{11}+s_{23}}{2}$
and providing no information yields a payoff of $\frac{5/6+s_{12}+s_{22}}{2}$.
Both values approach $5/12$ when $\delta\rightarrow0$. By continuity,
there exists $\delta^{2}>0$ such that for any $\delta<\delta^{2}$,
the partitional persuasion value is less than $\frac{5}{12}+\frac{\varepsilon}{2}$.

Next, we will construct a cheap-talk equilibrium that yields a payoff
close to $1/2$ for small $\delta$. Consider a strategy profile $(\sigma,\rho)$
with two on-path messages $m_{1},m_{2}$: $\sigma(m_{1}|\omega_{2})=\frac{\mu_{12}}{1-\mu_{12}}$,
$\sigma(m_{2}|\omega_{2})=\frac{1-2\mu_{12}}{1-\mu_{12}}$, $\sigma(m_{1}|\omega_{1})=1$;
$\rho(a_{1}|m_{1})=p$, $\rho(a_{2}|m_{1})=1-p$, $\rho(a_{3}|m_{2})=1$,
where $p=\frac{1/3+s_{23}-s_{22}}{2/3+s_{21}-s_{22}}$. Since $\delta<1/3$,
$p\in(0,1)$ is a well-defined probability.

The strategy profile satisfies R-BR because the posterior upon observing
$m_{1}$ is $\mu_{12}$, and upon observing $m_{2}$ is $1$. We now
show that the profile also satisfies S-BR. For type $\omega_{2}$
Sender, the expected payoff of sending message $m_{2}$ is $1/2+s_{23}$
and the expected payoff of sending message $m_{1}$ is $p(5/6+s_{21})+(1-p)(1/6+s_{22})=1/2+s_{23}$,
so type $\omega_{2}$ Sender is indifferent and has no incentive to
deviate. For type $\omega_{1}$ Sender, the expected payoff of sending
message $m_{2}$ is $0+s_{13}$ and the expected payoff of sending
message $m_{1}$ is $p(1/3+s_{11})+(1-p)(2/3+s_{12})$. As $\delta\rightarrow0$,
$p(1/3+s_{11})+(1-p)(2/3+s_{12})\rightarrow1/2$ and $0+s_{13}\rightarrow0$,
so type $\omega_{1}$ Sender strictly prefers to send message $m_{1}$.
By continuity, there exists $\delta^{3}>0$ such that for any $\delta<\delta^{3},$
the strategy profile $(\sigma,\rho)$ is a cheap-talk equilibrium,
yielding a value $\frac{1}{2}(1/2+s_{23})+\frac{1}{2}[p(1/3+s_{11})+(1-p)(2/3+s_{12})]>1/2-\frac{\varepsilon}{2}$.

Therefore, for any $\delta<\delta^{*}\equiv\min\{\delta^{1},\delta^{2},\delta^{3},\frac{1}{4}\}$,
and for any $(u_{S},u_{R})\in E^{\delta}$, $\Delta_{\mathcal{C}}\left(u_{S},u_{R}\right)<(\frac{1}{2}+\frac{\varepsilon}{2})-(\frac{1}{2}-\frac{\varepsilon}{2})=\varepsilon$,
and $\Delta_{\mathcal{R}}\left(u_{S},u_{R}\right)>(\frac{1}{2}-\frac{\varepsilon}{2})-(\frac{5}{12}+\frac{\varepsilon}{2})=\frac{1}{12}-\varepsilon$.

We have established that for any $\varepsilon>0$, there exists a
positive measure set of environments $E$ such that for any $(u_{S},u_{R})\in E$,
we have $\Delta_{\mathcal{C}}(u_{S},u_{R})<\varepsilon$ and $\Delta_{\mathcal{R}}(u_{S},u_{R})>\frac{1}{12}-\varepsilon$.
Now fix any $\hat{\varepsilon}>0$, let $\varepsilon=\min\{\hat{\varepsilon},\frac{1}{12}-\frac{1}{13}\}$,
our previous result implies that there exists a positive measure set
of environments $E$ such that for any $(u_{S},u_{R})\in E$, $\Delta_{\mathcal{C}}(u_{S},u_{R})<\varepsilon\leq\hat{\varepsilon}$
and $\Delta_{\mathcal{R}}(u_{S},u_{R})>\frac{1}{12}-\varepsilon\geq\frac{1}{12}-(\frac{1}{12}-\frac{1}{13})=\frac{1}{13}$.
\end{document}